\documentclass[a4paper,USenglish,cleveref,thm-restate]{lipics-v2021}

\hideLIPIcs
\pdfoutput=1

\nolinenumbers

\usepackage{tikz}
\usepackage{xspace}
\usepackage{amsmath,amssymb,amsthm}
\usepackage{bm}
\usepackage[normalem]{ulem}
\usepackage{graphicx}
\usepackage{makecell}
\usepackage{multirow}

\usepackage[vlined,boxed,linesnumbered]{algorithm2e}
\usepackage{cleveref} 

\newcommand{\ie}{\textit{i.e.},\xspace}
\newcommand{\eg}{\textit{e.g.},\xspace}
\newcommand{\etc}{\textit{etc.}\xspace}

\renewcommand{\paragraph}[1]{\smallskip\noindent \textbf{#1} \;}

\makeatletter
\@ifpackageloaded{algorithm2e}{%
\PackageInfo{cleveref}{`algorithm2e' support loaded}%
\crefalias{algocf}{algorithm}%
\crefname{line}{line}{lines} 
\crefalias{algocfline}{line}%
\crefalias{AlgoLine}{line}%
\let\cref@old@algocf@nl@sethref\algocf@nl@sethref%
\renewcommand{\algocf@nl@sethref}[1]{%
    \cref@old@algocf@nl@sethref{#1}%
    \cref@constructprefix{AlgoLine}{\cref@result}%
    \@ifundefined{cref@AlgoLine@alias}%
        {\def\@tempa{AlgoLine}}%
        {\def\@tempa{\csname cref@AlgoLine@alias\endcsname}}%
    \xdef\cref@currentlabel{%
        [\@tempa][\arabic{AlgoLine}][\cref@result]%
        \csname p@AlgoLine\endcsname\csname theAlgoLine\endcsname}}%
}{}
\makeatother

\let\oldnl\nl
\newcommand{\nonl}{\renewcommand{\nl}{\let\nl\oldnl}}

\newcommand{\commStep}{%
  \smallskip%
  {\SetKwComment{Comment}{}{}\Comment{\fbox{Communication step}}}
  \smallskip%
}

\newcommand{\compStep}{%
  \smallskip%
  {\SetKwComment{Comment}{}{}\Comment{\fbox{Computation step}}}
  \smallskip%
}

\DontPrintSemicolon

\SetKwHangingKw{InstParam}{instantiation parameter:}{}{}
\SetKwHangingKw{InstParams}{instantiation parameters:}{}{}
\SetKwHangingKw{Init}{initialization:}{}{}
\SetKwHangingKw{LocVar}{local variable of every process}{}{}
\SetKwProg{ProcCode}{code of every process}{ is}{}
\SetKwIF{If}{ElseIf}{Else}{if}{then}{else if}{else}{}

\SetKwComment{Comment}{$\rhd$}{}
\SetCommentSty{textit}

\newcommand{\cOr}{\textup{\textbf{or}}\xspace}
\newcommand{\cAnd}{\textup{\textbf{and}}\xspace}

\SetKwComment{Comment}{$\rhd$}{}
\SetCommentSty{textit}
\SetArgSty{upshape}
\SetProgSty{upshape}

\newcommand{\ttrue}{\ensuremath{\mathtt{true}}\xspace}
\newcommand{\ffalse}{\ensuremath{\mathtt{false}}\xspace}

\newcommand{\xor}{\ensuremath{\oplus}}

\newcommand{\wait}{\ensuremath{\mathsf{wait}}\xspace}
\newcommand{\prpick}{\ensuremath{\mathsf{pseudo\_random\_pick}}\xspace}

\newcommand{\communicate}{\ensuremath{\mathsf{communicate}}\xspace}
\newcommand{\observe}{\ensuremath{\mathsf{observe}}\xspace}
\newcommand{\observed}{\ensuremath{\mathsf{observed}}\xspace}

\newcommand{\ooutput}{\ensuremath{\mathsf{output}}\xspace}

\newcommand{\outputi}{\textsc{output}\xspace}
\newcommand{\initi}{\textsc{init}\xspace}
\newcommand{\proposei}{\textsc{propose}\xspace}

\newcommand{\nooutput}{\ensuremath{\mathit{no\_out}}\xspace}

\newcommand{\Async}{\ensuremath{\text{Async}}\xspace}
\newcommand{\Sync}{\ensuremath{\text{Sync}}\xspace}

\newcommand{\Vin}{\ensuremath{V_\text{in}}\xspace}
\newcommand{\vin}{\ensuremath{v_\text{in}}\xspace}
\newcommand{\Vout}{\ensuremath{V_\text{out}}\xspace}
\newcommand{\vout}{\ensuremath{v_\text{out}}\xspace}

\newcommand{\SV}{\ensuremath{\textit{OS}}\xspace}
\newcommand{\ST}{\ensuremath{\textit{SOS}}\xspace}
\newcommand{\STI}{\ensuremath{\textit{SOSI}}\xspace}
\newcommand{\FP}{\ensuremath{\textit{FP}}\xspace}
\newcommand{\CDP}{\ensuremath{\textit{CDP}}\xspace}
\newcommand{\fp}{\ensuremath{\textit{fp}}\xspace}
\newcommand{\cdp}{\ensuremath{\textit{cdp}}\xspace}
\newcommand{\seed}{\ensuremath{\textit{sd}}\xspace}

\newcommand{\information}{\ensuremath{\textit{I}}\xspace}

\newif\ifannote

\ifannote
    
    \newcommand{\anncomment}[3]{{\color{#1}[#2: #3]}}
\else
    
    \newcommand{\anncomment}[3]{}
\fi

\newcommand{\xOmit}[1]{}
\newcommand{\del}[1]{}

\title{Tight Conditions for Binary-Output Tasks \texorpdfstring{\\}{} under Crashes}

\titlerunning{Tight Conditions for Binary-Output Tasks under Crashes}

\author{Timoth\'e Albouy}{IMDEA Software Institute, Spain}{timothe.albouy@imdea.org}{https://orcid.org/0000-0001-9419-6646}{}

\author{Antonio Fern\'andez Anta}{IMDEA Software Institute and IMDEA Networks Institute, Spain}{antonio.fernandez@imdea.org}{https://orcid.org/0000-0001-6501-2377}{}

\author{Chryssis Georgiou}{University of Cyprus, Cyprus}{chryssis@ucy.ac.cy}{https://orcid.org/0000-0003-4360-0260}{}

\author{Nicolas Nicolaou}{Algolysis Ltd, Cyprus}{nicolas@algolysis.com}{https://orcid.org/0000-0001-7540-784X}{}

\author{Junlang Wang}{IMDEA Networks Institute, Spain}{junlang.wang@imdea.org}{https://orcid.org/0009-0003-6004-8823}{}

\authorrunning{Albouy, Fern\'andez Anta, Georgiou, Nicolaou, and Wang}

\Copyright{Timothé Albouy, Antonio Fernández Anta, Chryssis Georgiou, Nicolas Nicolaou, and Junlang Wang}

\keywords{
Distributed solvability, Asynchrony, Synchrony, Impossibility proofs, Binary-output tasks, Crash tolerance, Disagreement.
}

\funding{
This work has been partially supported by the Spanish Ministry of Science and Innovation under grants SocialProbing (TED2021-131264B-I00) and DRONAC (PID2022-140560OB-I00), the ERDF ``A way of making Europe'', NextGenerationEU, and the Spanish Government's ``Plan de Recuperaci\'on, Transformaci\'on y Resiliencia''.
}

\ccsdesc[100]{Theory of computation 
$\rightarrow$ Distributed algorithms}

\begin{document}

\maketitle

\begin{abstract}
This paper explores necessary and sufficient system conditions to solve distributed tasks with binary outputs (\textit{i.e.}, tasks with output values in $\{0,1\}$). We focus on the distinct output sets of values a task can produce (intentionally disregarding validity and value multiplicity), considering that some processes may output no value. In a distributed system with $n$ processes, of which up to $t \leq n$ can crash, we provide a complete characterization of the tight conditions on $n$ and $t$ under which every class of tasks with binary outputs is solvable, for both synchronous and asynchronous systems. This output-set approach yields highly general results: it unifies multiple distributed computing problems, such as binary consensus and symmetry breaking, and it produces impossibility proofs that hold for stronger task formulations, including those that consider validity, account for value multiplicity, or move beyond binary outputs.
\end{abstract}

\section{Introduction}

Distributed computing examines a wide range of coordination problems involving multiple processes interacting through a communication medium, such as a message-passing network or a shared memory.
Many of these problems can be abstracted as distributed tasks, which can be seen as black boxes taking an input vector and producing an output vector.
In these vectors, each process has (at most) one private input and one private output value.
For example, the well-known \textit{consensus} problem can be represented as a task with one input value per process (the proposals) and one output value per process (the decisions).
In addition to this input/output interface, a task definition can also impose requirements on the following aspects:
\begin{itemize}
    \item \textit{Input constraint}: a constraint on the set of possible input vectors (which is especially relevant for \textit{colored} tasks~\cite{HKR13} and asymmetric problems, \eg reliable broadcast~\cite{B87});

    \item \textit{Output constraint}: a constraint on the set of possible output vectors (\eg in consensus, all processes decide the same value);

    \item \textit{Validity}: connection between the inputs and outputs (\eg in consensus, the decided value must have been previously proposed by some process).
\end{itemize}
A central challenge in this field is to characterize the necessary and sufficient conditions under which distributed tasks can be solved, with many possible variations based on different assumptions on timing (asynchrony, synchrony) or failures (crashes, Byzantine faults, omissions).
The celebrated FLP theorem exemplifies this approach by demonstrating that consensus is impossible in any asynchronous system with even one process crash~\cite{FLP85}.

\paragraph{Towards a unifying framework for task solvability.}
When exploring the boundary between solvable and unsolvable tasks in distributed computing, a strategic approach is to analyze the weakest, most general possible version of a task for which an impossibility result still holds.
An impossibility proof for a weak task is particularly powerful because it automatically applies to all stronger versions of that same task.

Much of the literature on distributed computability has focused on particular tasks, such as consensus~\cite{FLP85}, renaming~\cite{CRR11}, set agreement~\cite{C93}, or election~\cite{KKM90}.
These studies have established powerful solvability and impossibility results, but they often rely on model-specific arguments (\eg on the communication medium) or on constraints such as validity, which tie output to input values.
This fragmentation makes it difficult to see the common computational structure underlying different tasks.
A natural question, therefore, is whether there exists a unifying perspective that abstracts away from inputs and instead focuses on the essential combinatorial structure of task outputs.

\paragraph{Our approach.}
In this paper, we address this question for the case of \textit{binary-output} tasks (\ie whose output values belong to $\{0,1\}$) in \textit{crash-prone} systems.
We introduce a classification based only on the \textit{sets of distinct output values} that executions of a binary-output task $T$ may produce, disregarding multiplicity in the output values.
More formally, we say that a single execution of $T$ produces an \textit{output set} included in $\{0,1\}$, considering that some processes may output no value (possibly intentionally).
There are only four possible binary output sets produced by an execution of $T$: $\varnothing$, $\{0\}$, $\{1\}$, and $\{0,1\}$.
Hence, the \textit{set of output sets} of $T$ is defined as the set $O \subseteq \{\varnothing, \{0\}, \{1\}, \{0,1\}\}$ of all possible output sets that can be produced across all the executions of $T$ (by providing all possible inputs).

This very weak abstraction allows us to strip tasks down to their minimally observable outcomes and ask: given the system size $n$ and crash tolerance $t$, which sets of output sets are achievable?
Or, conversely, given a set of output sets $O$, which combinations of $n$ and $t$ make $O$ achievable?
Our main result is a complete characterization of the tight conditions on $n$ and $t$ that determine solvability for all possible sets of output sets $O$, both in synchronous and asynchronous models.

The fact that our framework abstracts away the inputs and only focuses on the output sets of binary-output tasks, implies that impossibility proofs established are directly applicable to a broader range of stronger, more constrained tasks, including validity-based and multi-valued tasks.
This shift in focus does not just simplify the analysis; it also allows for the generalization of fundamental necessary conditions that apply to stronger task formulations.
Moreover, our framework remains expressive enough to capture a wide array of tasks, from classic ones like binary consensus to novel ones such as a new form of symmetry breaking.

It is important to note that, while our approach excels at generalizing impossibility results, the specific system conditions on $n$ and $t$ we identify in this paper are guaranteed to be tight only for binary-output tasks without validity restrictions.
Any additional constraint, such as a validity condition or a requirement on value multiplicity, may necessitate a stronger condition on the system parameters to achieve solvability.

\paragraph{Contributions and roadmap.}
Our main contributions are the following.
\begin{enumerate}
    \item \textit{A novel framework for studying binary-output tasks}: We introduce a new methodology striving to unify all distributed tasks with binary output values.
    More precisely, we focus on the sets of distinct output bits that can be produced by these tasks in crash-prone environments, abstracting away the input values, the output multiplicity, or the communication medium  (message passing or shared memory).
    As a result, the formalization that we obtain (see \Cref{sec:formal-pb}) is quite simple (it only relies on combinatorial techniques), thus facilitating formal reasoning, while staying expressive enough to capture many significant families of distributed tasks.

    \item \textit{A complete solvability characterization}: We exhaustively examine all 16 possible combinations of distinct output bits that binary-output tasks can produce, and we provide the tight conditions to implement each of them (see \Cref{sec:characterization}, \Cref{tab:characterization}), both in the asynchronous and synchronous cases, by proving both their necessity (see \Cref{sec:necessity}), via impossibility proofs, and their sufficiency (see \Cref{sec:sufficiency}), via algorithms and correctness proofs.

    \item \textit{New symmetry-breaking problems}: As an interesting twist, our novel way to classify binary-output tasks allowed us to discover new interesting problems, in particular, one that we baptized \textit{disagreement} (different from classical weak/strong symmetry breaking), which must always guarantee that the system does not agree on one single output value (see \Cref{sec:sufficiency}).
\end{enumerate}

\Cref{sec:related-work} exposes the research landscape in which this work is situated, and
concluding remarks are provided in \Cref{sec:conclusion}.
For the sake of presentation clarity, \Cref{apx:more-algos} contains omitted algorithms and proofs.

\section{Related Work} \label{sec:related-work}

Our work strives to classify and characterize many classes of distributed tasks.
In this section, we first present the tasks that are relevant to our characterization (along with some of their most salient (un)solvability results), and we then review the existing endeavors attempting to unify entire families of tasks under the same framework.

The two most studied (non-distinct) families of distributed tasks are \textit{agreement} and \textit{symmetry breaking}.
Informally, in agreement, we want to constrain the maximum number of different outputs, while in symmetry breaking we want to do the opposite, \ie constrain the minimum number of different outputs.

\paragraph{Agreement.}
The agreement family includes tasks such as $k$-set agreement~\cite{C93}, reliable broadcast~\cite{B87}, or consensus~\cite{L96}.
In particular, consensus is a fundamental agreement abstraction of distributed computing, where all participants propose a value and must eventually agree on one of the proposed values.
A foundational impossibility result for this problem is the FLP theorem~\cite{FLP85}, which shows that consensus is impossible in an asynchronous system with even one process crash.
In this context, Mostéfaoui, Rajsbaum, and Raynal explored how conditions on the input vector can make consensus solvable in asynchronous crash-prone systems~\cite{MRR03}.
In a way, we take the opposite approach: instead of characterizing the solvability of a task based on its \textit{input}, we focus on its \textit{output} (yielding necessary conditions that also hold if we account for inputs).
Aside from the asynchronous case, some other works also study the solvability of agreement tasks in synchronous systems~\cite{R02a,SW89}.

\paragraph{Symmetry breaking.}
The symmetry-breaking family includes tasks such as leader election~\cite{KKM90}, renaming~\cite{CRR11}, and weak/strong symmetry breaking (WSB/SSB)~\cite{IRR11}.
In leader election, only one process outputs $1$, while all others output $0$.
In $k$-renaming, $k$ processes in the system output different values (which correspond to their new identities).
Notice that our binary-output approach captures $2$-renaming, but not ${>}2$-renaming.
WSB guarantees that $0$ and $1$ are output only if all processes are correct, and SSB adds the property that $1$ is always output (even if there are faults).
Some solvability results also exist for symmetry breaking, \eg on the wait-free solvability of WSB and renaming~\cite{C10}.

\paragraph{Unification attempts.}
Several endeavors before us have attempted to unify multiple classes of distributed tasks under the same framework, in order to draw general results from their common structure.
For instance, \textit{generalized symmetry breaking} (\textit{GSB}) aims to unify many symmetry-breaking tasks \cite{CRR13,IRR11}, such as election, renaming, and weak symmetry-breaking.
Furthermore, many works employ a topological approach to study the computability of distributed tasks.
With topology, tasks are modeled as input and output complexes, and solvability is characterized by the existence of continuous maps, with impossibility linked to invariants such as connectivity.
One of the most prominent results in this area is the \textit{asynchronous computability theorem} (\textit{ACT}), which characterizes the solvability of wait-free tasks in an asynchronous shared-memory model~\cite{HS99}.
The \textit{musical benches} problem also leverages topological techniques to capture $2$-set agreement and $2$-renaming under the same abstraction in a wait-free shared-memory context \cite{GR05}.
For more material on the applications of combinatorial topology to distributed computing, we refer the interested reader to Herlihy, Kozlov, and Rajsbaum's monograph~\cite{HKR13}.

\section{Computing Model and Problem Formalization} \label{sec:formal-pb}

\subsection{Computing Model}


For clarity, we provide in \Cref{tab:notations} a list of concepts and notations used in this paper.

\begin{table}[t]
\centering
\caption{Concepts and notations used in this paper.}
\label{tab:notations}
\begin{tabular}{|c|c|}
    \hline
    \textbf{Concept or notation} & \textbf{Meaning} \\
    \hline\hline
    $p_i$ & process of the system with identity $i$ \\
    \hline
    $P$ & Set of processes in the system \\
    \hline
    $n$ & Number of processes in the system ($0 \le n = |P|$) \\
    \hline
    $t$ & Upper bound on the number of crashed processes ($0 \le t \le n$) \\
    \hline
    $f$ & Effective number of Byzantine processes in a run ($0 \leq f \leq t$) \\
    \hline
    $\sigma \in \{\Sync,\Async\}$ & Timing model (synchronous or asynchronous) \\
    \hline
    PRNG & Pseudo-random number generation \\
    \hline
    $\prpick(S)$ & Function returning a pseudo-random value $v$ from set $S$ \\
    \hline
    $T, A, E$ & Task, algorithm, execution \\
    \hline
    $\Vin \in (\mathbb{I} \cup \{\bot\})^n$ & Input vector \\
    \hline
    $\Vout \in (\mathbb{O} \cup \{\bot\})^n$ & Output vector, with $\mathbb{O}=\{0,1\}$ \\
    \hline
    $\bot$ & Sentinel value denoting no input/output \\
    \hline
    $\star$ & Unspecified value \\
    \hline
    $\nooutput$ & Boolean indicating if $\varnothing$ is a possible output set or not \\
    \hline
\end{tabular}
\end{table}

\paragraph{Process model.}
We consider a distributed system with a set $P=\{p_1, p_2, ..., p_n\}$ of processes ($|P| = n$). 
Processes are deterministic computing entities that take steps according to their local state and the events they observe, following an algorithm $A$.
Failures are restricted to \textit{crash faults}: in an execution of the system a process may halt prematurely and take no further steps, but it does not deviate from its algorithm before crashing.
We assume an upper bound $t$ on the number of processes that may crash in an execution, with $0 \leq t \leq n$.
A process that does not crash in an execution $E$ is said to be \textit{correct} in $E$.

\paragraph{Communication model.}
Processes interact through a generic communication medium.
This medium is \textit{reliable} as far as it does not suppress, duplicate, or corrupt information.
Every process $p \in P$ has access to two abstract operations that capture the behavior of this communication medium:
\begin{itemize}
    \item $\communicate$ $\information$: process $p$ disseminates some information $\information$ to the system processes;

    \item $\observe$ $\information$ (callback event): process $p$ is notified that information $\information$ was communicated.
\end{itemize}

From a terminology point of view, we say that processes \textit{communicate} and \textit{observe} information.
The medium guarantees that all correct processes eventually obtain a consistent view of the set of communicated information, while crashed processes may only have partial but always valid views.
More formally, the communication abstraction satisfies the following properties (the ``C'' prefix stands for ``communication'').
\begin{itemize}
    \item \textbf{C-Validity:} If a process $p$ observes information $\information$, then $\information$ must have been previously communicated by some process $p'$.

    \item \textbf{C-Local-Termination:} If a \textit{correct} process $p$ communicates information $\information$, then some \textit{correct} process $p'$ (if there is any) eventually observes $\information$.

    \item \textbf{C-Global-Termination:} If a process $p$ observes information $\information$, then all \textit{correct} processes eventually observe $\information$.
\end{itemize}

C-Validity is a safety property, while C-Local-Termination and C-Global-Termination are liveness properties.
These $\communicate/\observe$ operations can be implemented straightforwardly on top of classic communication media such as message-passing networks or shared memory, regardless of the number of crashes $t$ or the timing assumptions (synchronous, asynchronous, \etc).
For instance, in message passing, they can be realized using reliable broadcast~\cite{B87} (implementable under synchrony and asynchrony with any $t \leq n$), which ensures consistency of observations.
They can also be implemented from atomic shared memory\footnote{
    Atomic memory can itself be constructed from weaker memory models, such as safe or regular memory~\cite{R13}.
}: each process writes its communicated information into a new register, while other processes periodically scan all registers to detect new observable information.

\paragraph{Timing models: $\Async$ and $\Sync$.}
As usual, we assume the existence of a global clock, to which processes have no access.
We consider two classical timing models: asynchrony and synchrony, respectively denoted $\Async$ and $\Sync$.
For simplicity, we assume in both models that local computation occurs instantaneously, while only communication takes time.

In $\Async$, the speed of information propagation is arbitrary but positive: if a correct process communicates some information $\information$, then all correct processes eventually observe $\information$, but the delay before the observation may be unbounded.

In $\Sync$, communication and process execution proceed in globally coordinated lock-step rounds.
Each round consists of two steps: first, a \textit{communication step}, where processes communicate and observe information, and a \textit{computation step}, where processes can perform actions such as outputting values~\cite{R02a}.
Importantly, all information communicated during the communication step of a round is observed in that step, before the computation step of the same round begins.
Therefore, in the $\Sync$ model, the communication medium provides the following additional liveness property.

\begin{itemize}
    \item \textbf{C-Synchrony:} If an arbitrary process $p$ communicates some information $\information$ at the start of synchronous round $R$, then all processes that observe $\information$ do so during the communication step of round $R$.
\end{itemize}

\paragraph{Pseudo-randomness.}
The algorithms presented in this paper rely on pseudo-random number generation (PRNG).
Importantly, the use of PRNG does not strengthen the computing model: PRNG can be implemented deterministically given a local seed, and seeds can be derived from the inherent nondeterminism in distributed executions~\cite{KSW21}.
The correctness of our algorithms requires that the generated numbers are different across all executions (a perfectly uniform PRNG distribution is unnecessary).
For that, it is enough that local seeds change from one execution to another.
In the algorithms of this paper, we rely on the $\prpick(V)$ function that, given a set of values $V$, returns a value $v \in V$ pseudo-randomly.

\subsection{Problem Formalization} \label{sec:formalization}

\paragraph{Outputs.}
In an execution, a process $p \in P$ outputs a value $v \in \{0,1\}$ using the operation $\ooutput$ $v$.
In a given execution, each process $p \in P$ outputs \textit{at most one} value.
Hence, all its invocations of the $\ooutput$ operation must have the same value.

We often use the XOR operation to obtain the one's complement of a bit $v \in \{0,1\}$.
The logical formula we use is $1 \xor v$, which gives $1$ if $v=0$, and $0$ if $v=1$.

\paragraph{Tasks.}
We define a task $T$ as a relation between input vectors 
$\Vin = (\vin^1, ..., \vin^n)$ and output vectors 
$\Vout=(\vout^1, ..., \vout^n)$, where $\vin^i$ and $\vout^i$ are the input and output values of process $p_i$, respectively.
That is, we have $(\Vin,\Vout) \in T$ if and only if $\Vout$ is a possible output of the task when the input is $\Vin$.
The elements $v_\text{in}^i$ of an input vector $\Vin$ belong to the set $\mathbb{I} \cup \{\bot\}$, where $\mathbb{I}$ is the input alphabet and $\bot$ is a sentinel value (not in $\mathbb{I}$) that represents no input.
Similarly, the elements $v_\text{out}^i$ of an output vector $\Vout$ belong to the set $\mathbb{O} \cup \{\bot\}$, where $\mathbb{O}$ is the output alphabet and $\bot$ is a special value (not in $\mathbb{O}$) that represents no output.
In this work we consider only \textit{binary-output} tasks, \ie $\mathbb{O}=\{0,1\}$.
Observe that a task $T$ can have several possible outputs for the same input.
By abuse of notation, we denote the set of all outputs $\Vout$ of $T$ for an input $\Vin$ as $T(\Vin) = \{\Vout \mid (\Vin,\Vout) \in T\}$.

\paragraph{Output sets.}
In this work, we will not differentiate output vectors that contain the same set of different output values.
For that, we use the following notation for the set of distinct values in an output vector $\Vout$, which we call an \textit{output set}:
$\SV(\Vout) = \{\vout^i \in \Vout \mid \vout^i \neq \bot\}$.
Since we consider only binary-output tasks, it holds that 
$\SV(\Vout) \in 2^{\{0,1\}} = \{\varnothing,\{0\},\{1\},\{0,1\}\}$.

We then define the \textit{set of output sets} of a pair task/input vector $(T,\Vin)$ as the set of all possible output sets that can be produced by executing $T$ with $\Vin$:
$\STI(T,\Vin) = \{\SV(\Vout) \mid \Vout \in T(\Vin)\}$.
Observe that $\STI(T,\Vin) \subseteq \{\varnothing,\{0\},\{1\},\{0,1\}\}$, and that $\STI(T,\Vin) = \varnothing$ occurs when $T(\Vin) = \varnothing$, because there is no $\Vout$ such that $(\Vin,\Vout) \in T$.

Finally, we define the \textit{set of output sets} of a task $T$ as the set of all possible output sets that $T$ can produce with all possible input vectors $\Vin$:
$\ST(T) = \{\SV(\Vout) \mid \exists \Vin \in (\mathbb{I} \cup \{\bot\})^n: \Vout \in T(\Vin)\}$.
Observe again that $\ST(T) \subseteq \{\varnothing,\{0\},\{1\},\{0,1\}\}$.
Hence, all tasks $T$ can be classified into 16 classes, each characterized by which subset of $\{\varnothing,\{0\},\{1\},\{0,1\}\}$ the set of output sets $\ST(T)$ is.

\paragraph{Algorithms and executions.}
An \textit{algorithm} $A$ defines the steps taken by each process $p \in P$ as a function of its local state (including its input) and the communication observed.
The algorithm $A$ combined with input $\Vin$ is denoted $A(\Vin)$.
An \textit{execution} $E$ of $A(\Vin)$ is a (finite or infinite) sequence of steps taken by processes following algorithm $A$ with input $\Vin$, combined with the observation of communication events.
The \textit{output} of an execution $E$ of $A(\Vin)$ is the vector $\Vout$ with the output value that is produced in $E$ by each process $p$, or $\bot$ if no value is output by the process.

\paragraph{Failure and communication patterns.}
A \textit{failure pattern} defines in an execution the identity of faulty processes and the time instant each faulty process stops taking steps.
The possible failure patterns in an execution depend on $n$ and the value of $f \leq t$ when defined, which is the exact number of failures.
We denote the set of all failure patterns for a given pair $(n,f)$ by $\FP(n,f)$.

A \textit{communication-delay pattern} defines the delay experienced by each communication in the execution.
The possible communication-delay patterns in an execution depend on the synchrony ($\sigma \in \{\Async, \Sync\}$) of the system.
We denote the set of all communication-delay patterns for a given $\sigma$ by $\CDP(\sigma)$.

Since we only consider deterministic algorithms and local computation is instantaneous, an execution is fully characterized by the tuple $(A,\Vin,\seed,\fp,\cdp)$, where $A$ is the algorithm, $\Vin$ the input vector, $\seed$ is the seed of the PRNG, $\fp$ is the failure pattern, and $\cdp$ is the communication-delay pattern.

\paragraph{Implementation of a set of output sets.}
A \textit{system configuration} is a triple $C=(n,t,\sigma)$, where $n$ is the number of processes, $t$ is the maximum number of crashed processes in any execution $E$, and $\sigma \in \{\Async, \Sync\}$ defines the timing model of the distributed system in which $E$ runs.

\begin{definition}[Implementation of a set of output sets]
Algorithm $A$ \emph{implements} the non-empty set of output sets $O \subseteq \{\varnothing,\{0\},\{1\},\{0,1\}\}$ under system configuration $C=(n,t,\sigma)$ with $\sigma \in \{\Async, \Sync\}$ if for all $0 \leq f \leq t$, we have the following properties.
\begin{itemize}
    \item \textbf{Safety:} For all $\Vin$, all possible PRNG seeds $\seed$, all failure patterns $\fp \in \FP(n,f)$, and all communication-delay patterns $\cdp \in \CDP(\sigma)$, all executions $(A,\Vin,\seed,\fp,\cdp)$ of $A(\Vin)$ in a system with $n$ processes, $f$ failures, and timing model $\sigma$, have output vectors $\Vout$ such that $\SV(\Vout) \in O$.

    \item \textbf{Completeness:} For each $o \in O$, there is a $\Vin$, a PRNG seed $\seed$, a failure pattern $\fp \in \FP(n,f)$, and a communication-delay pattern $\cdp \in \CDP(\sigma)$, such that execution $(A,\Vin,s,\fp,\cdp)$ of $A(\Vin)$ has an output vector $\Vout$ where $\SV(\Vout) = o$.
\end{itemize}
\end{definition}

\begin{definition}[Binary-output task solvability]
Given a binary-output task $T$ with non-empty set of output sets $\ST(T)$, task $T$ can be solved under system configuration $C=(n,t,\sigma)$ with $\sigma \in \{\Async, \Sync\}$ only if there is an algorithm $A$ that implements $\ST(T)$ under that configuration.
\end{definition}
In the above definitions, we exclude the degenerate case where the set of output sets is empty.

\section{Tight Solvability Conditions for Binary-Output Tasks} \label{sec:characterization}

In this section, we introduce \Cref{tab:characterization}, which presents an exhaustive characterization of the tight conditions to implement all 16 possible binary sets of output sets, providing a necessary condition to solve their corresponding tasks.
The first column is the line index.
The second column corresponds to the set of output sets (the crossed sets in red are the forbidden output sets).
The third column specifies the timing model assumed ($\Async$, $\Sync$, or any model if the tightness proof does not depend on timing assumptions).
The fourth column provides the tight (\ie necessary and sufficient) system condition on $n$ (the system's size) and $t$ (the crash-resilience) to implement the given set of output sets under the corresponding timing model.
The fifth column refers to the necessity proof of the tight condition, while the sixth column refers to its sufficiency proof.

We can observe that the sets of output sets of many well-studied tasks are included in \Cref{tab:characterization}.
Line 10 corresponds to the set of output sets of classical binary consensus (\ie each execution outputs either $0$ or $1$)~\cite{FLP85}, while line 9 corresponds to abortable~\cite{BDFG03,C07} or randomized~\cite{B83} binary consensus (\ie some executions may never terminate and output a value).
Line 8 corresponds to $2$-renaming \cite{CRR11} (\ie renaming of only 2 processes in the system).
Line 2 corresponds to $2$-set agreement \cite{C93} (up to $2$ different values can be output by the system).
Line 4 (and its symmetric counterpart, line 6) corresponds to strong symmetry-breaking~\cite{AP16}: $1$ is always output, but $0$ can also be output in favorable conditions (\eg all processes participate, or there are no crashes).

Remark that we prove the necessity of the condition $t=0$ to solve the $\Async$ case of line 10 (implementing the set of output sets $\{\{0\},\{1\}\}$) by relying on the FLP theorem~\cite{FLP85}.
However, our model differs from the one in the original FLP paper; in particular, we do not consider consensus validity (\ie decisions are proposals), we do not require every process to output a value, and we do not assume a particular communication medium.
For this reason, we rely on \cite{AFGGNW24}, which presents another proof of the impossibility of asynchronous resilient consensus that does not make any of the aforementioned assumptions.

\newcommand{\impo}[1]{{
    \color{red}\ifmmode\text{\sout{\ensuremath{#1}}}\else\sout{#1}\fi
}}

\begin{table}[tb]
\caption{Characterization of the tight conditions for implementing the 16 sets of output sets.}
\label{tab:characterization}
\centering
\resizebox{.8\textwidth}{!}{
\begin{tabular}{|c|c|c|c|c|c|}
    \hline
    \textbf{\#} & \makecell{\textbf{Exact set of} \\  \textbf{output sets $\bm{O}$}} & \makecell{\textbf{Timing} \\ \textbf{model}} & \makecell{\textbf{Tight solvability} \\ \textbf{condition}} & \makecell{\textbf{Necessity} \textbf{proof}} & \makecell{\textbf{Sufficiency} \textbf{proof}} \\
    \hline
    1 & $\varnothing, \{0\}, \{1\}, \{0,1\}$ & \makecell{\Async \\ \& \Sync} & $n \ge t, n \ge 2$ & \Cref{obs:nec-proc-corr} & \makecell{\Cref{alg:async-all-uncond} / \Cref{thm:async-all-uncond} \\ with $V{=}\{0,1,\bot\}$} \\
    \hline
    2 & $\impo{\varnothing}, \{0\}, \{1\}, \{0,1\}$ & \makecell{\Async \\ \& \Sync} & $n > t, n \ge 2$ & \Cref{obs:nec-proc-corr} & \makecell{\Cref{alg:async-all-uncond} / \Cref{thm:async-all-uncond} \\ with $V{=}\{0,1\}$} \\
    \hline
    3 & $\varnothing, \impo{\{0\}}, \{1\}, \{0,1\}$ & \makecell{\Async \\ \& \Sync} & $n \ge t, n \ge 2$ & \Cref{obs:nec-proc-corr} & \makecell{\Cref{alg:async-all-cond} / \Cref{thm:async-all-cond} \\ with $v{=}1$, $\nooutput{=}\ttrue$} \\
    \hline
    4 & $\impo{\varnothing}, \impo{\{0\}}, \{1\}, \{0,1\}$ & \makecell{\Async \\ \& \Sync} & $n > t, n \ge 2$ & \Cref{obs:nec-proc-corr} & \makecell{\Cref{alg:async-all-cond} / \Cref{thm:async-all-cond} \\ with $v{=}1$, $\nooutput{=}\ffalse$} \\
    \hline
    5 & $\varnothing, \{0\}, \impo{\{1\}}, \{0,1\}$ & \makecell{\Async \\ \& \Sync} & $n \ge t, n \ge 2$ & \Cref{obs:nec-proc-corr} & \makecell{\Cref{alg:async-all-cond} / \Cref{thm:async-all-cond} \\ with $v{=}0$, $\nooutput{=}\ttrue$} \\
    \hline
    6 & $\impo{\varnothing}, \{0\}, \impo{\{1\}}, \{0,1\}$ & \makecell{\Async \\ \& \Sync} & $n > t, n \ge 2$ & \Cref{obs:nec-proc-corr} & \makecell{\Cref{alg:async-all-cond} / \Cref{thm:async-all-cond} \\ with $v{=}0$, $\nooutput{=}\ffalse$} \\
    \hline
    \multirow{2}{*}[-.8em]{7} & \multirow{2}{*}[-.8em]{$\varnothing, \impo{\{0\}}, \impo{\{1\}}, \{0,1\}$} & \Async & $n > \frac{3}{2}t+1, n \ge 2$ & \Cref{obs:nec-proc-corr} \& \Cref{thm:nec-01} & \makecell{\Cref{alg:async-sym-break} / \Cref{thm:async-sym-break} \\ with $\nooutput{=}\ttrue$} \\
    \cline{3-6}
     & & \Sync & $n \ge t+2, n \ge 2$ & \Cref{obs:nec-proc-corr} \& \Cref{thm:nec-min-corr} & \makecell{\Cref{alg:sync-sym-break} / \Cref{thm:sync-sym-break} \\ with $\nooutput{=}\ttrue$} \\
    \hline
    \multirow{2}{*}[-.8em]{8} & \multirow{2}{*}[-.8em]{$\impo{\varnothing}, \impo{\{0\}}, \impo{\{1\}}, \{0,1\}$} & \Async & $n > \frac{3}{2}t+1, n \ge 2$ & \Cref{obs:nec-proc-corr} \& \Cref{thm:nec-01} & \makecell{\Cref{alg:async-sym-break} / \Cref{thm:async-sym-break} \\ with $\nooutput{=}\ffalse$} \\
    \cline{3-6}
     & & \Sync & $n \ge t+2, n \ge 2$ & \Cref{obs:nec-proc-corr} & \makecell{\Cref{alg:sync-sym-break} / \Cref{thm:sync-sym-break} \\ with $\nooutput{=}\ffalse$} \\
    \hline
    9 & $\varnothing, \{0\}, \{1\}, \impo{\{0,1\}}$ & \makecell{\Async \\ \& \Sync} & $n \ge t, n \ge 1$ & \Cref{obs:nec-proc-corr} & \makecell{\Cref{alg:async-single} / \Cref{thm:async-single} \\ with $\nooutput=\ttrue$} \\
    \hline
    \multirow{2}{*}[-.2em]{10} & \multirow{2}{*}[-.2em]{$\impo{\varnothing}, \{0\}, \{1\}, \impo{\{0,1\}}$} & \Async & $t = 0, n \ge 1$ & \Cref{obs:nec-proc-corr} \& \cite{AFGGNW24,FLP85} & \makecell{\Cref{alg:async-single} / \Cref{thm:async-single} \\ with $\nooutput{=}\ffalse$} \\
    \cline{3-6}
     & & \Sync & $n > t, n \ge 1$ & \Cref{obs:nec-proc-corr} & \Cref{alg:sync-consensus} / \Cref{thm:sync-consensus} \\
    \hline
    11 & $\varnothing, \impo{\{0\}}, \{1\}, \impo{\{0,1\}}$ & \makecell{\Async \\ \& \Sync} & $n \ge t, n \ge 1$ & \Cref{obs:nec-proc-corr} & \makecell{\Cref{alg:async-all-uncond} / \Cref{thm:async-all-uncond} \\ with $V{=}\{1,\bot\}$} \\
    \hline
    12 & $\impo{\varnothing}, \impo{\{0\}}, \{1\}, \impo{\{0,1\}}$ & \makecell{\Async \\ \& \Sync} & $n > t, n \ge 1$ & \Cref{obs:nec-proc-corr} & \makecell{\Cref{alg:async-all-uncond} / \Cref{thm:async-all-uncond} \\ with $V{=}\{1\}$} \\
    \hline
    13 & $\varnothing, \{0\}, \impo{\{1\}}, \impo{\{0,1\}}$ & \makecell{\Async \\ \& \Sync} & $n \ge t, n \ge 1$ & \Cref{obs:nec-proc-corr} & \makecell{\Cref{alg:async-all-uncond} / \Cref{thm:async-all-uncond} \\ with $V{=}\{0,\bot\}$} \\
    \hline
    14 & $\impo{\varnothing}, \{0\}, \impo{\{1\}}, \impo{\{0,1\}}$ & \makecell{\Async \\ \& \Sync} & $n > t, n \ge 1$ & \Cref{obs:nec-proc-corr} & \makecell{\Cref{alg:async-all-uncond} / \Cref{thm:async-all-uncond} \\ with $V{=}\{0\}$} \\
    \hline
    15 & $\varnothing, \impo{\{0\}}, \impo{\{1\}}, \impo{\{0,1\}}$ & \makecell{\Async \\ \& \Sync} & $n \ge t, n \ge 0$ & \Cref{obs:nec-proc-corr} & \makecell{\Cref{alg:async-all-uncond} / \Cref{thm:async-all-uncond} \\ with $V{=}\{\bot\}$} \\
    \hline
    16 & $\impo{\varnothing}, \impo{\{0\}}, \impo{\{1\}}, \impo{\{0,1\}}$ & \makecell{\Async \\ \& \Sync} & N/A & N/A & N/A \\
    \hline
\end{tabular}
}
\end{table}

\section{Necessity: Impossibility Proofs} \label{sec:necessity}

In this section, we prove the necessity of the tightness conditions of \Cref{tab:characterization}, by showing the impossibility of implementing the corresponding sets of output sets when these conditions are not satisfied.

\begin{observation} \label{obs:nec-proc-corr}
Under any timing model ($\Async$ or $\Sync$), the following conditions are necessary to implement a set of output sets $O$:
$$ n \geq \max(\{|o| : o \in O\}), \hspace{1em} \text{ and } \hspace{1em} n-t \geq \min(\{|o| : o \in O\}). $$
\end{observation}

\begin{proof}
Let us consider a set of output sets $O$.
Condition $n \geq \max(\{|o| : o \in O\})$ is necessary for implementing $O$, otherwise there would not be enough processes to output all the values of the largest $o \in O$, violating completeness.
Condition $n-t \geq \min(\{|o| : o \in O\})$ is also necessary for implementing $O$.
Otherwise, when $f=t$, there would not be enough processes guaranteed to stay correct to output all the values of the smallest $o \in O$, violating safety.
\end{proof}

\begin{theorem} \label{thm:nec-min-corr}
Under any timing model ($\Async$ or $\Sync$), the condition $n-t \geq 2$ is necessary for implementing the set of output sets $O = \{\varnothing,\{0,1\}\}$.
\end{theorem}

\begin{proof}
We prove the lemma by contradiction: assume that there is an algorithm $A$ that can implement $O = \{\varnothing,\{0,1\}\}$ under $n-t < 2$.
This means that at least $n-1$ processes can crash during an execution.

Let us consider an execution $E_c$ of $A$ that produces output set $\{0,1\}$.
This execution must exist for completeness.
In particular, let us consider the first process $p \in P$ that outputs some value $v \in \{0,1\}$ in $E_c$ with respect to the global time, and let $\tau$ be the time when $p$ outputs $v$.\footnote{
  If multiple processes simultaneously output a value at time $\tau$, we pick any of these processes, as their respective outputs cannot have physically influenced each other.
}
Then, it is possible to craft an execution $E$ that is exactly the same as $E_c$ up to time $\tau$, at which time all processes except $p$ crash before outputting anything, and $p$ is correct.
In execution $E$, the output set is $\{v\}$, which violates safety: contradiction.
\end{proof}

\begin{theorem} \label{thm:nec-01}
Under $\Async$, condition $n > \frac{3}{2} t+1$ is necessary for implementing the sets of output sets $O = \{\{0,1\}\}$ and $O' = \{\varnothing, \{0,1\}\}$.
\end{theorem}

\begin{proof}
By contradiction, let us assume that there exists an algorithm $A$ that implements the set of output sets $\{\{0,1\}\}$ or $\{\varnothing,\{0,1\}\}$ under $\Async$ and $n \leq \frac{3}{2}t+1$.
Note that $n \geq 2$, from \Cref{obs:nec-proc-corr}, which implies that $t \geq 1$.

We first prove that, no matter the case (whether $A$ implements $\{\{0,1\}\}$ or $\{\varnothing,\{0,1\}\}$), $A$ must have a crash-free execution $E_0$ producing the output set $\{0,1\}$.
\begin{itemize}
    \item If $A$ implements the set of output sets $\{\{0,1\}\}$, then any crash-free execution of $A$ can only output set $\{0,1\}$.

    \item If $A$ implements the set of output sets $\{\varnothing,\{0,1\}\}$, then we consider any execution $E_c$ of $A$ that outputs set $\{0,1\}$ but that may present crashes.
    Let us denote by $P_d$ the set of processes that output any value in $E_c$, and by $P_c$ the set of processes that crash in $E_c$.
    We can construct an execution $E_0$ of $A$ that is the same as $E_c$, except that it has no crash, and all information communicated by processes of $P_c$ after the point where they would crash in $E_c$ is delayed due to asynchrony, such that it is only observed by processes of $P_d$ after they produce their outputs.
    Therefore, the outputs of the processes of $P_d$ cannot have been influenced by the fact that the processes in $P_c$ are correct, and $E_0$ must also produce the output set $\{0,1\}$.
\end{itemize}

We therefore refer to $E_0$ as a crash-free execution of $A$ producing the output set $\{0,1\}$.
Let $p_1 \in P$ be the first process that outputs some value $v_1 \in \{0,1\}$ in $E_0$ w.r.t. global time (at time $\tau_1$).
Let execution $E_1$ be an execution that is the same as $E_0$ up to $\tau_1$, but where $p_1$ crashes immediately after it outputs $v_1$.

We now construct a sequence of executions $E_2,...,E_t$ of $A$ inductively as follows.
Let us consider some $i = 2,...,t$.
To construct $E_i$, we assume by induction hypothesis that, in execution $E_{i-1}$ (starting with $E_1$), process $p_1$ outputs value $v_1$, therefore $E_{i-1}$ must necessarily produce the output set $\{0,1\}$ (it cannot produce the $\varnothing$ output set anymore).
Hence, there must be another process in the system that outputs the opposite value $1 \xor v_1$ in $E_{i-1}$.
In particular, there must be some process $p_i \in P$ which is the second one (after $p_1$) to output a value $v_i \in \{0,1\}$ in $E_{i-1}$ (w.r.t. global time) at time $\tau_i$.
We then make execution $E_i$ the same as $E_{i-1}$ until time $\tau_i$, but we make process $p_i$ crash just before it outputs a value, leaving $p_1$ be the only process that outputs a value in $E_i$ up to time $\tau_i$. 

When we reach execution $E_t$, there have been $t$ process crashes, so we cannot crash processes anymore.
Finally, we denote by $p_{t+1}$ the second process that outputs some value in $E_t$ at time $\tau_{t+1}$, w.r.t. global time.

We now create a new crash-free execution $E_\text{crash-free}$ that is the same as $E_t$, except that all processes $p_1,...,p_t$ do not crash, but all communication that they make after their respective output is delayed until $\tau_(t+1)$.
Then, for any process $p \in P \setminus \{p_1,...,p_t\}$, prior to $\tau_{t+1}$, $E_\text{crash-free}$ is indistinguishable from $E_t$ from the point of view of $p$.
In other words, in $E_\text{crash-free}$, the processes of $\{p_1,...,p_t\}$ are indistinguishable from crashed processes to the other processes, before $\tau_{t+1}$.

Let us consider the set of processes $\{p_1,...,p_{t+1}\}$ in $E_\text{crash-free}$.
We can divide this set into two partitions: the subset $P_0$ of processes that output $0$ in $E_\text{crash-free}$, and the subset $P_1$ of processes that output $1$ in $E_\text{crash-free}$.
Let us denote by $P_\text{min}$ (resp., $P_\text{maj}$) the smaller (resp., bigger) set between $P_0$ and $P_1$ (if they have the same size, then $P_\text{min}$ is $P_0$ and $P_\text{maj}$ is $P_1$).
Remark that $|P_\text{min}|+|P_\text{maj}| = t+1$ and $|P_\text{min}| \leq \frac{t+1}{2} \leq |P_\text{maj}|$.
We then denote by $P_?$ the subset of processes that are not in $P_\text{min}$ or $P_\text{maj}$: $P_? = P \setminus \{p_1,...,p_{t+1}\}$.
We have $|P_?| = n-t-1 \leq \frac{3}{2}t+1-t-1 = \frac{t}{2}$.
Moreover, we have $|P_\text{min}|+|P_?| \leq \frac{t+1}{2}+\frac{t}{2} = t+\frac{1}{2}$.
But as $|P_\text{min}|$, $|P_?|$, and $t$ are all integers, then we also have $|P_\text{min}|+|P_?| \leq t$.

\begin{figure}[t]
\centering
\begin{tikzpicture}
\colorlet{red}{red!70!black}
\tikzset{cross/.style={red, line width=.5pt}}
\tikzset{dashed_arrow/.style={dashed, -latex, line width=0.5pt}}
\tikzset{text_node/.style={align=center, font=\footnotesize}}
\tikzset{set_circle/.style={draw, line width=0.5pt}}

\coordinate (Pmin_pos) at (-3, 0);
\coordinate (Pquestion_pos) at (0, 0);
\coordinate (Pmaj_pos) at (3.2, 0);
\coordinate (out_v_pos) at (6.5, 0);

\draw (Pmin_pos) ellipse [x radius=1, y radius=0.5];
\node[text_node] at (Pmin_pos) {$P_{\text{min}}$\\${\le}\frac{t+1}{2}$ procs.};
\draw[cross] (-3.55, -0.55) -- (-2.45, 0.55);
\draw[cross] (-3.55, 0.55) -- (-2.45, -0.55);

\draw (Pquestion_pos) ellipse [x radius=1, y radius=0.5];
\node[text_node] at (Pquestion_pos) {$P_?$ \\ ${\le}\frac{t}{2}$ procs.};
\draw[cross] (-0.55, -0.55) -- (0.55, 0.55);
\draw[cross] (-0.55, 0.55) -- (0.55, -0.55);

\draw (Pmaj_pos) ellipse [x radius=1.1, y radius=0.5];
\node[text_node] (Pmaj) at (Pmaj_pos) {$P_{\text{maj}}$ \\ ${\ge}\frac{t+1}{2}$ procs.};

\node[text_node, draw=none, inner xsep=0.5em] (out-v) at (out_v_pos) {all output \\ same $v$};
\draw[dashed_arrow] (4.3, 0) -- (out-v);
\end{tikzpicture}

\caption{Execution $E_{\text{crash}}$, where processes of $P_{\text{min}}$ and $P_?$ crash, and only processes of $P_{\text{maj}}$ stay correct, which leads to a safety violation.}
\label{fig:nec-01}
\end{figure}
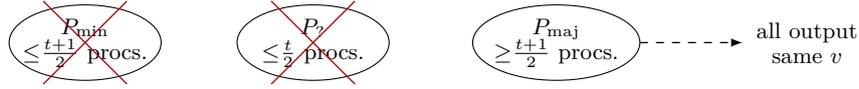

As illustrated in \Cref{fig:nec-01}, we can finally construct another execution $E_\text{crash}$, which is the same as $E_\text{crash-free}$ up to time $\tau_{t+1}$, except that we make all processes of $P_\text{min}$ and $P_?$ crash in $E_\text{crash}$ immediately before their output step (as $|P_\text{min}|+|P_?| \leq t$).
After time $\tau_{t+1}$, the only remaining correct processes in $E_\text{crash}$ are those of $P_\text{maj}$.
The processes of $P_\text{maj}$ all output the same value $v$ ($v=0$ if $P_\text{maj}=P_0$, $v=1$ if $P_\text{maj}=P_1$), thus producing the output set ${v}$.
Hence, safety is violated, and we have a contradiction.
\end{proof}

\section{Sufficiency: Implementing Algorithms} \label{sec:sufficiency}

In this section, we provide algorithms, along with their correctness proofs, implementing the binary sets of output sets of \Cref{tab:characterization}.
For modularity and conciseness, the algorithms of this paper feature \textit{instantiation parameters}, which allow each one of them to have several possible behaviours, and thus implement multiple sets of output sets in \Cref{tab:characterization}.

For the sake of presentation clarity, some algorithms and proofs used in \Cref{tab:characterization} are presented in detail in \Cref{apx:more-algos}.
This section focuses on the algorithms implementing lines 7 and 8 of \Cref{tab:characterization} (\Cref{alg:async-sym-break} for the $\Async$ case, \Cref{alg:sync-sym-break} for the $\Sync$ case), which we deem the most interesting ones in our classification.
We call them \textit{disagreement} algorithms, since they must always guarantee that the system never outputs only one single value ($0$ or $1$).

\subsection{Asynchronous Disagreement Algorithm (Lines 7-8)}

The algorithm of this section (\Cref{alg:async-sym-break}) addresses the $\Async$ case of lines 7 and 8 in \Cref{tab:characterization}.
It is instantiated by providing a Boolean $\nooutput \in \{\ttrue, \ffalse\}$ (\cref{line:async-sym-break:instantiation}), which determines if the empty output set $\varnothing$ can be produced or not.

At the initialization of the algorithm, the set of processes $P$ is partitioned into three subsets, $P_0,P_1,P_?$ (\cref{line:async-sym-break:AlgInitiation}), which respectively contain at least $\frac{t+1}{2}$, $\frac{t}{2}$, and $\frac{t+1}{2}$ processes.
By combining this with the system assumption, we obtain that the union of any two of these subsets contains at least one correct process.
Moreover, a set of at least $t+1$ processes $P_\text{init} \subset P$ is also selected (which is only used when $\nooutput=\ttrue$).

Every process $p_\text{init} \in P_\text{init}$ passes the condition at \cref{line:async-sym-break:initMsg} only if $\varnothing$ is not a forbidden output set ($\nooutput = \ttrue$) and it pseudo-randomly picks $0$.
If it does, $p_\text{init}$ communicate $\initi$.
Every process $p_v \in P_v, v \in \{0,1\}$ first checks if $\nooutput = \ttrue$ (\cref{line:async-sym-break:wait-init}).
If it does, $p_v$ waits until it observes some $\initi$ information.
If $p_v$ passes the previous line, it will then output $v$ (\cref{line:async-sym-break:out-v}) and communicate $\outputi(v)$ (\cref{line:async-sym-break:comm-v}).
Every process $p_? \in P_?$ first waits until it observes some $\outputi(v)$ (\cref{line:async-sym-break:wait-v}).
If it passes this $\wait$ statement, then $p_?$ outputs the opposite of $v$, namely $1 \xor v$ (\cref{line:async-sym-break:out-inv}).

Before proving the correctness of \Cref{alg:async-sym-break}, we first show some intermediary results.

\begin{algorithm}[tb]
\InstParam{Boolean $\nooutput \in \{\ttrue,\ffalse\}$.} \label{line:async-sym-break:instantiation}
\smallskip

\Init{pick three subsets of processes $P_0,P_1,P_? \subset P, |P_0| \geq \frac{t+1}{2},$ \nonl \\
  \hspace{1em} $|P_1| \geq \frac{t}{2}, |P_?| \geq \frac{t+1}{2}, P_0 \cup P_1 \cup P_? = P, P_0 \cap P_1 = P_1 \cap P_? = P_0 \cap P_? = \varnothing$; \nonl \\
  \hspace{1em} pick some set of processes $P_\text{init} \subset P, |P_\text{init}| \geq t+1$.} \label{line:async-sym-break:AlgInitiation}
\smallskip

\ProcCode{$p_\text{init} \in P_\text{init}$}{ \label{line:async-sym-break:startPinit}
    \lIf{$\nooutput = \ttrue$ \cAnd $\prpick(\{0,1\})=0$}{$\communicate$ $\initi$.} \label{line:async-sym-break:initMsg}
}
\smallskip

\ProcCode{$p_v \in P_v, v \in \{0,1\}$}{ \label{line:async-sym-break:startP01}
    \lIf{$\nooutput = \ttrue$}{$\wait$ until $p_v$ $\observed$ some $\initi$;} \label{line:async-sym-break:wait-init}
    $\ooutput$ $v$; \label{line:async-sym-break:out-v} \\
    $\communicate$ $\outputi(v)$. \label{line:async-sym-break:comm-v}
}

\ProcCode{$p_? \in P_?$}{ \label{line:async-sym-break:startPinv}
    $\wait$ until $p_?$ $\observed$ some $\outputi(v)$ for the first time; \label{line:async-sym-break:wait-v} \\
    $\ooutput$ $1 \xor v$. \label{line:async-sym-break:out-inv}
}

\caption{Disagreement asynchronous algorithm for the $\Async$ case of lines 7-8 of \Cref{tab:characterization}, assuming $\frac{3}{2}t+1 < n \geq 2$.}
\label{alg:async-sym-break}
\end{algorithm}

\begin{observation} \label{obs:some-init-corr}
Since $P_0$ and $P_1$ are disjoint, we have $|P_0 \cup P_1| \geq \frac{t+1}{2} + \frac{t}{2} = t+\frac{1}{2}$.
Moreover, as the left-hand side is an integer, then we also have $|P_0 \cup P_1| \geq t+1$, hence there is at least one correct process $p_v^c \in P_0 \cup P_1$.
\end{observation}

\begin{lemma}\label{lem:01-if-single-special}
If some process $p \in P_?$ outputs a value $v \in \{0,1\}$ at \cref{line:async-sym-break:out-inv}, then another process $p' \in P_{1 \xor v}$ also outputs the opposite value $1 \xor v$ at \cref{line:async-sym-break:out-v}.
\end{lemma}

\begin{proof}
If some $p \in P_?$ outputs $v$ at \cref{line:async-sym-break:out-inv}, then it must have observed some $\outputi(1 \xor v)$ at \cref{line:async-sym-break:wait-v}.  
By C-Validity, this $\outputi(1 \xor v)$ was communicated by some process $p' \in P_{1 \xor v}$ at \cref{line:async-sym-break:comm-v}, after it had output $1 \xor v$ at \cref{line:async-sym-break:out-v}.
\end{proof}

\begin{lemma} \label{lem:corr-pass-wait}
If a process $p \in P_0 \cup P_1$ outputs some value $v \in \{0,1\}$ at \cref{line:async-sym-break:out-v}, then all correct processes in $P_0 \cup P_1$ pass \cref{line:async-sym-break:wait-init}.
\end{lemma}

\begin{proof}
Let us assume that some process $p \in P_0 \cup P_1$ outputs a value $v \in \{0,1\}$.
We now show that all correct processes in $P_0 \cup P_1$ pass \cref{line:async-sym-break:wait-init}.
\begin{itemize}
  \item If $\nooutput = \ffalse$, then this is immediate.  
  \item If $\nooutput = \ttrue$, then all correct processes in $P_0 \cup P_1$ reach the $\wait$ statement at \cref{line:async-sym-break:wait-init}.
  However, since $p$ reached \cref{line:async-sym-break:out-v}, it must have passed that $\wait$, thus observing some $\initi$.  
  By C-Global-termination, all correct processes in $P_0 \cup P_1$ also observe $\initi$ and therefore pass \cref{line:async-sym-break:wait-init}. \qedhere
\end{itemize}
\end{proof}

\begin{lemma} \label{lem:async-01-if-single}
If some process $p \in P$ outputs some value $v \in \{0,1\}$, then another process $p' \in P \setminus \{p\}$ also outputs the opposite value $1 \xor v$.
\end{lemma}

\begin{proof}
Let us assume that some process $p \in P$ outputs a value $v \in \{0,1\}$.
If $p \in P_?$, then \Cref{lem:01-if-single-special} states that the execution must produce the output set $\{0,1\}$ and we are done.
Otherwise, if $p \in P_0 \cup P_1$, then \Cref{lem:corr-pass-wait} applies, and all correct processes in $P_0 \cup P_1$ pass \cref{line:async-sym-break:wait-init}.
We consider the following two opposite cases.
\begin{enumerate}
    \item Case \textit{(i)}: At least one correct $p' \in P_?$ exists.
    As all correct processes in $P_0 \cup P_1$ pass \cref{line:async-sym-break:wait-init}, then, in particular, correct process $p_v^c \in P_0 \cup P_1$ (\Cref{obs:some-init-corr}) reaches \cref{line:async-sym-break:comm-v} and communicates some $\outputi(\star)$.
    By C-Local-termination and C-Global-termination, $p'$ will eventually observe this $\outputi(\star)$ information, and will therefore pass the $\wait$ statement at \cref{line:async-sym-break:wait-v} (either after it observed the $\outputi(\star)$ communicated by $p_v^c$, or because it observed some prior $\outputi(\star)$).
    Finally, $p'$ will reach \cref{line:async-sym-break:out-inv} and output some value $v' \in \{0,1\}$.
    \Cref{lem:01-if-single-special} therefore applies, and some other process $p'' \in P_{1 \xor v'}$ also outputs the opposite value $1 \xor v$.
    
    \item Case \textit{(ii)}: All processes in $P_?$ crash.
    In this case, at least $\lceil\frac{t+1}{2}\rceil$ crashes have already occurred on all processes of $P_?$, and the maximum number of remaining crashes is $t - \lceil\frac{t+1}{2}\rceil$.
    \begin{itemize}
        \item If $t$ is even, then $\lceil\frac{t+1}{2}\rceil=\frac{t+2}{2}=\frac{t}{2}+1$ and $t-\lceil\frac{t+1}{2}\rceil = t - \frac{t}{2} - 1 = \frac{t}{2}-1 < \frac{t}{2}$.
        \item If $t$ is odd, then $\lceil\frac{t+1}{2}\rceil = \frac{t+1}{2}$ and $t - \lceil\frac{t+1}{2}\rceil = t - \frac{t+1}{2} = \frac{t-1}{2} < \frac{t}{2}$.
    \end{itemize}
    Therefore, it is not possible to also crash all processes of $P_0$ or $P_1$, since they both contain at least $\frac{t}{2}$ processes.
    This implies that there are at least two correct processes $p_0 \in P_0,p_1 \in P_1$.
    Since all correct processes in $P_0 \cup P_1$ pass \cref{line:async-sym-break:wait-init} (\Cref{lem:corr-pass-wait}), then $p_0$ and $p_1$ both reach \cref{line:async-sym-break:out-v}, and respectively output $0$ and $1$.
\end{enumerate}
Therefore, in every case, if some value $v \in \{0,1\}$ is output by a process $p \in P$, the opposite value $1 \xor v$ is also output by another process $p' \in P \setminus \{p\}$.
\end{proof}

\begin{theorem} \label{thm:async-sym-break}
Under $\Async$ and $\frac{3}{2}t + 1 < n \geq 2$, \Cref{alg:async-sym-break} instantiated with Boolean $\nooutput \in \{\ttrue,\ffalse\}$ implements the following sets of output sets:
$$ O = \begin{cases}
  \{\{0,1\}\} & \text{if } \nooutput = \ffalse, \\
  \{\varnothing, \{0,1\}\} & \text{if } \nooutput = \ttrue.
\end{cases} $$
\end{theorem}

\begin{proof}
We first show that condition $\frac{3}{2}t + 1 < n \geq 2$ is sufficient to construct the sets of processes used in \Cref{alg:async-sym-break}, namely $P_0,P_1,P_?,P_\text{init}$.
\begin{itemize}
    \item If $t$ is even, then $\frac{3}{2}t$ is an integer and $n > \frac{3}{2}t + 1$ implies $n \geq \frac{3}{2}t + 2$.
    Also, $\frac{t}{2}$ is an integer.
    Hence, there are enough processes to set $|P_0| \geq \frac{t}{2} + 1$, $|P_1| \geq \frac{t}{2}$, and $|P_?| \geq \frac{t}{2} + 1$, thus satisfying the conditions of \cref{line:async-sym-break:AlgInitiation}.  
    This implies that $|P_0| + |P_1| + |P_?| = \frac{3}{2}t + 2$, as desired.
    
    \item If $t$ is odd, then $\frac{3t-1}{2}$ is an integer and $\frac{3}{2}t=\frac{3t-1}{2} + \frac{1}{2}$. Then $n > \frac{3}{2}t + 1$ implies $n \geq \frac{3t-1}{2} + 2$.
    Also, $\frac{t-1}{2}$ is an integer.  
    Hence, there are enough processes to set $|P_0| \geq \frac{t-1}{2} + 1$, $|P_1| \geq \frac{t+1}{2}$, and $|P_?| \geq \frac{t-1}{2} + 1$, thus satisfying the conditions of \cref{line:async-sym-break:AlgInitiation}.  
    This implies that $|P_0| + |P_1| + |P_?| = \frac{3t-1}{2} + 2$, as desired.
\end{itemize}
Furthermore, since we have $n > \frac{3}{2}t + 1 \geq t+1$ and $|P_\text{init}| \geq t+1$, then there are also enough processes to construct the set $P_\text{init}$.  
Note also that $n \geq 2$ is already implied by $n > \frac{3}{2}t + 1$.
We consider the following two opposite cases.  
\begin{itemize}
  \item Case 1: $\nooutput = \ffalse$. We must show $O = \{\{0,1\}\}$.  
  \item Case 2: $\nooutput = \ttrue$. We must show $O = \{\varnothing, \{0,1\}\}$.
\end{itemize}
We now prove that the following safety and completeness results hold for both Cases 1 and 2.
\begin{enumerate}
    \item General-Safety: Impossibility of $\{0\}$ and $\{1\}$ in Cases 1 and 2.
    By corollary of \Cref{lem:async-01-if-single}, the output sets $\{0\}$ and $\{1\}$ cannot happen.

    \item General-Completeness: Possibility of $\{0,1\}$ in Cases 1 and 2.
    Let us first show that there is an execution where the correct process $p_v^c \in P_0 \cup P_1$ (\Cref{obs:some-init-corr}) passes \cref{line:async-sym-break:wait-init}.
    \begin{itemize}
        \item In Case 1, we have $\nooutput=\ffalse$, so $p_v^c$ passes this line.

        \item In Case 2, we have $\nooutput=\ttrue$, so $p_v^c$ reaches the $\wait$ statement at \cref{line:async-sym-break:wait-init}.
        However, since we have $|P_\text{init}| \geq t+1$, there is at least one correct process $p_\text{init}^c \in P_\text{init}$, and there exists an execution where $p_\text{init}^c$ picks $0$ and then communicates $\initi$ at \cref{line:async-sym-break:initMsg}, thus unlocking $p_v^c$ at \cref{line:async-sym-break:wait-init}.
    \end{itemize}
    Therefore, there is an execution where $p_v^c$ passes \cref{line:async-sym-break:wait-init}, and then outputs some value $v \in \{0,1\}$ at \cref{line:async-sym-break:out-v}.
    By \Cref{lem:async-01-if-single}, another process also outputs the opposite value $1 \xor v$, thus producing the output set $\{0,1\}$.
\end{enumerate}

The safety of Case 2 follows from General-Safety.
For Case 1, we also have to prove the impossibility of $\varnothing$.
Since $\nooutput=\ffalse$ in Case 1, correct process $p_v^c$ (\Cref{obs:some-init-corr}) passes the condition at \cref{line:async-sym-break:wait-init} and outputs some $v \in \{0,1\}$ at \cref{line:async-sym-break:out-v}, and thus $\varnothing$ is impossible. 

The completeness of Case 1 follows from General-Completeness.
For Case 2, we also have to prove the possibility of $\varnothing$.
Since $\nooutput=\ttrue$ in Case 2, there is an execution where no process from $P_\text{init}$ communicates $\initi$ at \cref{line:async-sym-break:initMsg} (because all processes in $P_\text{init}$ either crashed or picked $1$ at \cref{line:async-sym-break:initMsg}).
In this execution, no process in $P_0 \cup P_1$ passes the $\wait$ statement at \cref{line:async-sym-break:wait-init}, and no process in $P_?$ passes the $\wait$ instruction at \cref{line:async-sym-break:wait-v}.
Therefore, no process outputs any value in this execution, thus producing output set $\varnothing$.
\end{proof}

\subsection{Synchronous Disagreement Algorithm (Lines 7-8)}

The algorithm of this section (\Cref{alg:sync-sym-break}) addresses the $\Sync$ case of lines 7 and 8 in \Cref{tab:characterization}.
It is instantiated by providing a Boolean $\nooutput \in \{\ttrue, \ffalse\}$ (\cref{line:sync-sym-break:instantiation}), which determines if the empty $\varnothing$ output set can be produced or not.

At the initialization of the algorithm, all processes of $P$ are divided into two totally-ordered sequences $S_0$ and $S_1$, of $\lfloor\frac{n}{2}\rfloor$ and $\lceil\frac{n}{2}\rceil$ processes respectively, such that $|S_0|+|S_1|=n$.
Informally, the value $v \in \{0,1\}$ of a sequence $S_v$ corresponds to the \textit{default value} that the processes of $S_v$ output if they did not observe that another process already output this value before.
The index $i$ of every process $p_i^v$ in a sequence $S_v$ determines the round numbers in which $p_i^v$ will output a value (round $i$) and communicate information (round $i+1$).
Moreover, a set of at least $t+1$ processes $P_\text{init} \subset P$ is also selected (which is only used when $\nooutput=\ttrue$).
Each process $p_i^v \in S_v$ also has a local variable $v_\text{chosen}$, initialized to the sentinel value $\bot$, which stores the output value of $p_i^v$ between round $i$ and $i+1$. 

Every process $p_\text{init} \in P_\text{init}$ passes the condition at \cref{line:sync-sym-break:comm-init} only if $\varnothing$ is not a forbidden output set ($\nooutput = \ttrue$) and it pseudo-randomly picks $0$.
If it does, $p_\text{init}$ communicate $\initi$.
Every process $p_i^v \in S_v$ ($v \in \{0,1\}$) only performs actions during synchronous rounds $i$ and $i+1$ (if $i < \lceil\frac{n}{2}\rceil$), because of the conditions at \cref{line:sync-sym-break:comm-v,line:sync-sym-break:cond}.

During synchronous round $i$, $p_i^v$ does not pass the condition at \cref{line:sync-sym-break:comm-v} during the communication step, but it can enter the condition at \cref{line:sync-sym-break:cond}, if $\nooutput=\ffalse$ or if it observed some $\initi$.
If it enters the condition, then it updates its $v_\text{chosen}$ local variable: $v_\text{chosen}$ stores the opposite of $p_i^v$'s default value $v$ if it observed that another process in the system already output its $v$, otherwise $v_\text{chosen}$ stores $v$ (\cref{line:sync-sym-break:choose-v}).
Finally, $p_i^v$ outputs $v_\text{chosen}$ (\cref{line:sync-sym-break:out}).

During synchronous round $i+1$ (if $i < \lceil\frac{n}{2}\rceil$), $p_i^v$ passes the condition at \cref{line:sync-sym-break:comm-v} if it has output a value in the previous round (\ie $v_\text{chosen} \neq \bot$).
If it does, it communicates its output value through some $\outputi(v_\text{chosen})$ information (\cref{line:sync-sym-break:comm-v}).

Before proving the correctness of \Cref{alg:sync-sym-break}, we first show some intermediary results.

\begin{algorithm}[t]
\InstParam{Boolean $\nooutput \in \{\ttrue,\ffalse\}$.} \label{line:sync-sym-break:instantiation}
\smallskip

\Init{pick 2 distinct sequences of processes $S_0 = (p_1^0,...,p_{\lfloor\frac{n}{2}\rfloor}^0) \in P^{\lfloor\frac{n}{2}\rfloor}$, \nonl \\
    \hspace{1em} $S_0 = (p_1^1,...,p_{\lceil\frac{n}{2}\rceil}^1) \in P^{\lceil\frac{n}{2}\rceil}$ s.t $\forall p \in P: (p \in S_0 \land p \notin S_1) \lor (p \notin S_0 \land p \in S_1)$; \nonl \\
    \hspace{1em} pick some set of processes $P_\text{init} \subset P, |P_\text{init}| \geq t+1$.} \label{line:sync-sym-break:init}
\smallskip

\LocVar{$p_i^v \in S_v, v \in \{0,1\}$\textbf{:} value $v_\text{chosen} \gets \bot$.}
\smallskip

\ProcCode{$p_\text{init} \in P_\text{init}$ \textbf{at synchronous round} $R=1$}{
    \lIf{$\nooutput=\ttrue$ \cAnd $\prpick(\{0,1\})=0$}{$\communicate$ $\initi$.} \label{line:sync-sym-break:comm-init}
}
\smallskip

\ProcCode{$p_v \in S_v, v \in \{0,1\}$ \textbf{at synchronous round} $R \in [1..\lceil\frac{n}{2}\rceil]$}{
    \commStep
    \lIf{$R = i+1$ \cAnd $p_i^v$ has output some value}{$\communicate$ $\outputi(v_\text{chosen})$.} \label{line:sync-sym-break:comm-v}
    \compStep
    \If{$R = i$ \cAnd ($\nooutput=\ffalse$ \cOr $p_i^v$ $\observed$ $\initi$)}{ \label{line:sync-sym-break:cond}
        $v_\text{chosen} \gets 1 \xor v$ \textbf{if} $p_i^v$ $\observed$ some $\outputi(v)$ \textbf{else} $v$; \label{line:sync-sym-break:choose-v} \\
        $\ooutput$ $v_\text{chosen}$. \label{line:sync-sym-break:out}
    }
}
\caption{Disagreement synchronous algorithm for the $\Sync$ case of lines 7-8 of \Cref{tab:characterization}, assuming $t+2 \leq n \geq 2$.}
\label{alg:sync-sym-break}
\end{algorithm}

\begin{lemma} \label{lem:2-corr-in-S}
There must be two distinct correct processes $p_i^v \in S_v, p_j^{v'} \in S_{v'}$ for $v,v' \in \{0,1\}$ and $i,j \in [1..\lceil\frac{n}{2}\rceil]$.
\end{lemma}

\begin{proof}
Since $n \geq t+2$, then there are at least two correct processes $p,p' \in P$.
Moreover, by definition, all processes $p'' \in P$ belong to $S_0$ or $S_1$, so $p,p'$ belong to either $S_0$ or $S_1$, and their indices in these sequences are comprised between $1$ and $\lceil\frac{n}{2}\rceil$.
\end{proof}

\begin{lemma} \label{lem:corr-pass-if-out}
If some process $p \in P$ outputs a value $w \in \{0,1\}$ at \cref{line:sync-sym-break:out}, then two correct processes $p_i^v \in S_v, p_j^{v'} \in S_{v'}$ (where $i,j \in [1..\lceil\frac{n}{2}\rceil]$ and $v,v' \in \{0,1\}$) also output some values $v_1,v_2 \in \{0,1\}$ at \cref{line:sync-sym-break:out}.
\end{lemma}

\begin{proof}
Assume some process $p \in P$ outputs a value $w \in \{0,1\}$ at \cref{line:sync-sym-break:out} during synchronous round $k \in [1..\lceil\frac{n}{2}\rceil]$.
Let us also consider the two correct processes $p_i^v,p_j^{v'}$ belonging to either $S_0$ or $S_1$ (\Cref{lem:2-corr-in-S}), where $i,j \in [1..\lceil\frac{n}{2}\rceil]$ are round numbers and $v,v' \in \{0,1\}$ are default output values.
Process $p$ must have satisfied the condition at \cref{line:sync-sym-break:cond}, which leads to two opposite cases.
\begin{itemize}
    \item If $\nooutput=\ffalse$, then $p_i^v,p_j^{v'}$ do the following during round number $i$ and $j$, respectively: they enter the condition at \cref{line:sync-sym-break:cond}, reach \cref{line:sync-sym-break:out} and output some value in $\{0,1\}$.
    
    \item If $\nooutput=\ttrue$, then $p$ has observed some $\initi$ information, which, by C-Validity, must have been communicated during round $1$ by some process $p_\text{init} \in P_\text{init}$ at \cref{line:sync-sym-break:comm-init}.
    Moreover, by C-Global-termination and C-Synchrony, all correct processes in $P$ must also observe some $\initi$ information by the end of round $1$.
    In particular, correct processes $p_i^v,p_j^{v'}$ respectively satisfy the condition at \cref{line:sync-sym-break:cond} at rounds $i,j$, reach \cref{line:sync-sym-break:out}, and output some values $v_1,v_2 \in \{0,1\}$.
\end{itemize}
Therefore, $p_i^v,p_j^{v'}$ always output some values $v_1,v_2 \in \{0,1\}$ at \cref{line:sync-sym-break:out}.
\end{proof}

\begin{lemma} \label{lem:v-if-inv}
If some process $p_j^{v'} \in S_{v'}$ outputs the opposite of its default value $1 \xor v' \in \{0,1\}$ at \cref{line:sync-sym-break:out} during round $j \in [2..\lceil\frac{n}{2}\rceil]$, then some other process $p_k^\star \in P$ must have output $v'$ in a previous round $k < j$.
\end{lemma}

\begin{proof}
If a process $p_j^{v'} \in S_{v'}$ outputs the opposite of its default value $1 \xor v' \in \{0,1\}$ at \cref{line:sync-sym-break:out} during round $j \in [2..\lceil\frac{n}{2}\rceil]$, then it must have chosen $v_\text{chosen} = 1 \xor v'$ at \cref{line:sync-sym-break:choose-v}, and therefore it must have observed some $\outputi(v')$ information.
By C-Validity, this $\outputi(v')$ must have been communicated by some process $p_k^\star \in P$ at \cref{line:sync-sym-break:comm-v} at the beginning of round $k+1$ (where $2 \leq k+1 \leq j$).
Moreover, $p_k^\star$ must have output $v'$ at \cref{line:sync-sym-break:out} during round $k$ (where $1 \leq k < j$).
\end{proof}

\begin{lemma} \label{lem:sync-01-if-single}
If some process $p \in P$ outputs a value $v \in \{0,1\}$ at \cref{line:sync-sym-break:out}, then another process $p' \in P \setminus \{p\}$ also outputs the opposite value $1 \xor v$ at \cref{line:sync-sym-break:out}.
\end{lemma}

\begin{proof}
Let us assume that some process $p \in P$ outputs a value $v \in \{0,1\}$ at \cref{line:sync-sym-break:out}.
By \cref{lem:corr-pass-if-out}, two correct processes $p_i^v,p_j^{v'}$ belonging to either $S_0$ or $S_1$ respectively output some values $v_1,v_2 \in \{0,1\}$ at \cref{line:sync-sym-break:out}, during synchronous rounds $i,j \in [1..\lceil\frac{n}{2}\rceil]$.

If either $p_i^v$ or $p_j^{v'}$ output the opposite of their default value (\ie if $v_1 = 1 \xor v$ or $v_2 = 1 \xor v'$), then \cref{lem:v-if-inv} applies, and there is another process $p_k^\star$ that outputs $v_3$, which is the opposite value that they output (\ie $v_3 = 1 \xor v_1$ if $v_1 = 1 \xor v$, or $v_3 = 1 \xor v_2$ if $v_2 = 1 \xor v'$), so we are done.
Therefore, in the following, we assume that $p_i^v$ and $p_j^{v'}$ output their default values, \ie $v_1 = v$ and $v_2 = v'$.
We consider the following two cases.
\begin{itemize}
    \item Case \textit{(i)}: $i=j$.
    In this case, $p_i^v$ and $p_j^{v'}$ both output during the same synchronous round $i=j$.
    As $p_i^v$ and $p_j^{v'}$ are distinct processes, then they cannot belong to the same sequence $S_0$ or $S_1$, and we must have $v \neq v'$.
    But since $p_i^v$ and $p_j^{v'}$ respectively output $v$ and $v'$, then they output opposite values.

    \item Case \textit{(ii)}: $i < j$ (without loss of generality).
    In this case, $p_i^v$ has output before $p_j^{v'}$, $i \in [1..\lceil\frac{n}{2}\rceil-1]$, and $j \in [2..\lceil\frac{n}{2}\rceil]$.
    Let us remark that correct process $p_i^v$ must have communicated $\outputi(v_1)$ at \cref{line:sync-sym-break:comm-v} during round $i+1$ (where $2 \leq i+1 \leq j$), after it has ouput $v_1$ at \cref{line:sync-sym-break:out} during round $i$.
    By C-Local-termination, C-Global-termination, C-Synchrony, $p_j^{v'}$ must have observed $\outputi(v_1)$ by the end of round $i+1 \leq j$.
    However, since $v_2 = v'$ (\ie $p_j^{v'}$ outputs its default value $v'$), by \cref{line:sync-sym-break:choose-v}, $p_j^{v'}$ has not observed any $\outputi(v')$ information during or before round $j \geq 2$.
    This necessarily means that $v_1 \neq v' = v_2$, and therefore that $p_i^v$ and $p_j^{v'}$ output opposite values (\ie $v_1 = 1 \xor v_2$).
\end{itemize}
Whether in Case \textit{(i)} or Case \textit{(ii)}, processes $p_i^v$ and $p_j^{v'}$ output opposite values at \cref{line:sync-sym-break:out} (\ie $v_1 = 1 \xor v_2$), which concludes the lemma.
\end{proof}

\begin{theorem} \label{thm:sync-sym-break}
Under $\Sync$ and $t+2 \leq n \geq 2$, \Cref{alg:sync-sym-break} implements the following sets of output sets:
$$ O = \begin{cases}
    \{\{0,1\}\} & \text{if } \nooutput=\ffalse, \\
    \{\varnothing,\{0,1\}\} & \text{if } \nooutput=\ttrue.
\end{cases} $$
\end{theorem}

\begin{proof}
We first show that condition $t+2 \leq n \geq 2$ is sufficient to construct the sequences and set of processes used in \Cref{alg:async-sym-break}: $S_0,S_1,P_\text{init}$.
Since we have $n = \lfloor\frac{n}{2}\rfloor + \lceil\frac{n}{2}\rceil = |S_0|+|S_1|$, then there are enough processes to construct the sequences of distinct processes $S_0,S_1$.
Furthermore, since we have $n \geq t+2$ and $|P_\text{init}| \geq t+1$, then there are enough processes to construct the set $P_\text{init}$.
Let us also remark that $n \geq 2$ is already implied by $n \geq t+2$.
We consider the following two opposite cases.  
\begin{itemize}
    \item In Case 1, we have $\nooutput=\ffalse$.
    In this case, we must prove that the set of output sets is $O = \{\{0,1\}\}$.

    \item In Case 1, we have $\nooutput=\ttrue$.
    In this case, we must prove that the set of output sets is $O = \{\varnothing,\{0,1\}\}$.
\end{itemize}
We now prove that the following safety and completeness results hold for both Cases 1 and 2.
\begin{enumerate}
    \item General-Safety: Impossibility of $\{0\}$ and $\{1\}$ in Cases 1 and 2.
    By corollary of \Cref{lem:sync-01-if-single}, the output sets $\{0\}$ and $\{1\}$ cannot happen.

    \item General-Completeness: Possibility of ${0,1}$ in Cases 1 and 2.
    By \Cref{lem:2-corr-in-S}, there is a correct process $p_i^v$ in $S_0$ or $S_1$ that reaches the condition at \cref{line:sync-sym-break:cond} during synchronous round $i \in [1..\lceil\frac{n}{2}\rceil]$.
    Let us show that, no matter the case, there is an execution where $p_i^v$ passes the condition at \cref{line:sync-sym-break:cond}.
    \begin{itemize}
        \item In Case 1, we have $\nooutput=\ffalse$, so $p_i^v$ satisfies the condition \cref{line:sync-sym-break:cond} at round $i$.
        
        \item In Case 2, we have $\nooutput=\ttrue$.
        Since we have $|P_\text{init}| \geq t+1$, there is at least one correct process $p_\text{init}^c \in P_\text{init}$, and there exists an execution where $p_\text{init}^c$ picks $0$ and then communicates $\initi$ at \cref{line:sync-sym-break:init} during round $1$.
        By C-Local-Termination, C-Global-termination, and C-Synchrony, $p_i^v$ observes $\initi$ by the end of round $1$, and therefore, when it reaches round $i \geq 1$, the condition at \cref{line:sync-sym-break:cond} is satisfied.
    \end{itemize}
  Therefore, there is an execution where $p_i^v$ passes \cref{line:sync-sym-break:cond}, and then outputs some value $v \in \{0,1\}$ at \cref{line:sync-sym-break:out} during round $i$.
  By \Cref{lem:sync-01-if-single}, another process also outputs the opposite value $1 \xor v$, thus producing the output set $\{0,1\}$.
\end{enumerate}

The safety of Case 2 follows from General-Safety.
For Case 1, we also have to prove the impossibility of $\varnothing$.
By \Cref{lem:2-corr-in-S}, there is a correct process $p_i^v$ belonging to $S_0$ or $S_1$.
As $\nooutput=\ttrue$ in Case 2, then $p_i^v$ must pass \cref{line:sync-sym-break:cond} during round $i \in [1..\lceil\frac{n}{2}\rceil]$, reach \cref{line:sync-sym-break:out}, and output some value $v' \in \{0,1\}$.
Therefore, the output set $\varnothing$ is impossible.

The completeness of Case 1 follows from General-Completeness.
For Case 2, we have to prove the possibility of $\varnothing$.
Since $\nooutput=\ttrue$ in Case 2, there is an execution where no process from $P_\text{init}$ communicates $\initi$ at \cref{line:sync-sym-break:init} (because all processes in $P_\text{init}$ either crashed or picked $1$ at \cref{line:sync-sym-break:init}).
In this execution, no process in $P$ passes \cref{line:sync-sym-break:cond}, and therefore, no process outputs any value in this execution, thus producing output set $\varnothing$.
\end{proof}

\section{Conclusion} \label{sec:conclusion}

In this work, we exhaustively characterize solvability conditions for binary-output tasks under crash failures.
More particularly, we focus on the sets of distinct output values that executions of these tasks can produce, which breaks down binary-output tasks into 16 classes.
Our results cover both necessity, via impossibility proofs, and sufficiency, via implementing algorithms, thereby offering a definitive picture~(\Cref{tab:characterization}) of the boundary between possible and impossible.
By abstracting away from specific input assumptions and focusing on possible outputs, the conditions we present also serve as lower bounds (albeit not necessarily tight) for any stronger problem formulation.

Some of the results connect back to some well-studied problems, such as binary consensus and symmetry breaking, as discussed in \Cref{sec:characterization}.
One particularly interesting problem we discovered is the \textit{disagreement} problem, an instance of binary symmetry breaking, which requires the system to never ``agree'' on one single output value ($0$ or $1$).
To our knowledge, this particular problem has not been studied before, and it illustrates how our framework exposes subtle but fundamental problems that are not captured by classical formulations.

We propose in the following some possible extensions to our current framework.
Currently, we do not require that all correct processes eventually output some value.
It is left for future work to explore whether adding this requirement will change the conditions.
Exploring the tight conditions in partial synchronous environments is another interesting venue.
Another logical extension would be to consider multi-valued outputs (\ie not limited to $0$/$1$), which will significantly increase the number of possible sets of output sets.
Moreover, output sets hide some structural information of output values, such as their multiplicity or the information about which process outputs the value.
This can be done by adding additional information to output sets or by directly studying output vectors.
Finally, the task input and validity constraints (relation between input and output) are also aspects to be explored in future work.

\bibliographystyle{plain}
\bibliography{bibliography}

\begin{thebibliography}{10}

\bibitem{AFGGNW24}
Timoth{\'{e}} Albouy, Antonio~Fern{\'{a}}ndez Anta, Chryssis Georgiou, Mathieu Gestin, Nicolas Nicolaou, and Junlang Wang.
\newblock {AMECOS:} a modular event-based framework for concurrent object specification.
\newblock In {\em Proc. 28th Int'l Conference on Principles of Distributed Systems (OPODIS'24)}, volume 324 of {\em LIPIcs}, pages 4:1--4:29. Schloss Dagstuhl - Leibniz-Zentrum f{\"{u}}r Informatik, 2024.

\bibitem{AP16}
Hagit Attiya and Ami Paz.
\newblock Counting-based impossibility proofs for set agreement and renaming.
\newblock {\em J. Parallel Distributed Comput.}, 87:1--12, 2016.

\bibitem{B83}
Michael Ben{-}Or.
\newblock Another advantage of free choice: Completely asynchronous agreement protocols ({Extended} abstract).
\newblock In {\em Proc. 2nd {ACM} Symposium on Principles of Distributed Computing (PODC'83)}, pages 27--30. {ACM}, 1983.

\bibitem{BDFG03}
Romain Boichat, Partha Dutta, Svend Fr{\o}lund, and Rachid Guerraoui.
\newblock Deconstructing {Paxos}.
\newblock {\em {SIGACT} News}, 34(1):47--67, 2003.

\bibitem{B87}
Gabriel Bracha.
\newblock Asynchronous {Byzantine} agreement protocols.
\newblock {\em Inf. Comput.}, 75(2):130--143, 1987.

\bibitem{CRR11}
Armando Casta{{n}}eda, Sergio Rajsbaum, and Michel Raynal.
\newblock The renaming problem in shared memory systems: An introduction.
\newblock {\em Comput. Sci. Rev.}, 5(3):229--251, 2011.

\bibitem{CRR13}
Armando Casta{\~{n}}eda, Sergio Rajsbaum, and Michel Raynal.
\newblock Agreement via symmetry breaking: On the structure of weak subconsensus tasks.
\newblock In {\em Proc. 27th {IEEE} Int'l Symposium on Parallel and Distributed Processing (IPDPS'13)}, pages 1147--1158. {IEEE} Computer Society, 2013.

\bibitem{C93}
Soma Chaudhuri.
\newblock More choices allow more faults: Set consensus problems in totally asynchronous systems.
\newblock {\em Inf. Comput.}, 105(1):132--158, 1993.

\bibitem{C07}
Wei Chen.
\newblock Abortable consensus and its application to probabilistic atomic broadcast.
\newblock Technical report, Microsoft Research Asia, 2007.

\bibitem{FLP85}
Michael~J. Fischer, Nancy~A. Lynch, and Mike Paterson.
\newblock Impossibility of distributed consensus with one faulty process.
\newblock {\em J. {ACM}}, 32(2):374--382, 1985.

\bibitem{GR05}
Eli Gafni and Sergio Rajsbaum.
\newblock Musical benches.
\newblock In {\em Proc. 19th Int'l Conference on Distributed Computing (DISC'05)}, volume 3724 of {\em Lecture Notes in Computer Science}, pages 63--77. Springer, 2005.

\bibitem{HKR13}
Maurice Herlihy, Dmitry~N. Kozlov, and Sergio Rajsbaum.
\newblock {\em Distributed Computing Through Combinatorial Topology}.
\newblock Morgan Kaufmann, 2013.

\bibitem{HS99}
Maurice Herlihy and Nir Shavit.
\newblock The topological structure of asynchronous computability.
\newblock {\em J. {ACM}}, 46(6):858--923, 1999.

\bibitem{IRR11}
Damien Imbs, Sergio Rajsbaum, and Michel Raynal.
\newblock The universe of symmetry breaking tasks.
\newblock In {\em Proc. 18th Int'l Colloquium on Structural Information and Communication Complexity (SIROCCO'11)}, volume 6796 of {\em Lecture Notes in Computer Science}, pages 66--77. Springer, 2011.

\bibitem{KSW21}
Peter Kietzmann, Thomas~C. Schmidt, and Matthias W{\"{a}}hlisch.
\newblock A guideline on pseudorandom number generation {(PRNG)} in the {IoT}.
\newblock {\em {ACM} Comput. Surv.}, 54(6):112:1--112:38, 2022.

\bibitem{KKM90}
Ephraim Korach, Shay Kutten, and Shlomo Moran.
\newblock A modular technique for the design of efficient distributed leader finding algorithms.
\newblock {\em {ACM} Trans. Program. Lang. Syst.}, 12(1):84--101, 1990.

\bibitem{L96}
Nancy~A. Lynch.
\newblock {\em Distributed Algorithms}.
\newblock Morgan Kaufmann, 1996.

\bibitem{MRR03}
Achour Most{\'{e}}faoui, Sergio Rajsbaum, and Michel Raynal.
\newblock Conditions on input vectors for consensus solvability in asynchronous distributed systems.
\newblock {\em J. {ACM}}, 50(6):922--954, 2003.

\bibitem{R02a}
Michel Raynal.
\newblock Consensus in synchronous systems: A concise guided tour.
\newblock In {\em Proc. 9th Pacific Rim Int'l Symposium on Dependable Computing (PRDC'02)}, pages 221--228. {IEEE} Computer Society, 2002.

\bibitem{R13}
Michel Raynal.
\newblock {\em Concurrent Programming - Algorithms, Principles, and Foundations}.
\newblock Springer, 2013.

\bibitem{C10}
Armando~Castañeda Rojano.
\newblock {\em A study of the wait-free solvability of weak symmetry breaking and renaming}.
\newblock PhD thesis, Universidad Nacional Autónoma de México, 2010.

\bibitem{SW89}
Nicola Santoro and Peter Widmayer.
\newblock Time is not a healer.
\newblock In {\em Proc. 6th Annual Symposium on Theoretical Aspects of Computer Science (STACS'89)}, volume 349 of {\em Lecture Notes in Computer Science}, pages 304--313. Springer, 1989.

\end{thebibliography}

\appendix

\section{Sufficiency: Additional Algorithms and Proofs} \label{apx:more-algos}

This appendix presents the remaining algorithms and correctness proofs used for proving the sufficiency of the tightness conditions of \Cref{tab:characterization}.

\subsection{Communication-Less All-Output Algorithm (Lines 1-2, 11-15)}

The algorithm of this section (\Cref{alg:async-all-uncond}) addresses lines 1 to 2, and 11 to 15 in \Cref{tab:characterization}.
It requires no communication, as it only involves one process outputting a pseudo-random value in isolation.
Therefore, this algorithm works both in the $\Sync$ and $\Async$ timing models.
This algorithm is instantiated with one parameter (\cref{line:async-all-uncond:init}): a set of values $V \subseteq \{0,1,\bot\}$ determining all output behaviors of the algorithm.
Every process $p \in P$ only picks a pseudo-random value $v \in V$ (\cref{line:async-all-uncond:pick}) and outputs this value if it is not $\bot$ (\cref{line:async-all-uncond:out}).

\begin{theorem} \label{thm:async-all-uncond}
Under any timing model ($\Async$ or $\Sync$) and $n \geq \max(\{|o| : o \in O\})$, \Cref{alg:async-all-uncond} instantiated with the nonempty set of values $V \subseteq \{0,1,\bot\}$ implements the following sets of output sets:
$$ O = \begin{cases}
  2^V \setminus \{\varnothing\} &\text{if } \bot \notin V \text{ and } t < n, \\
  2^{(V \setminus \{\bot\})} &\text{if } \bot \in V \text{ and } t \leq n.
\end{cases} $$
\end{theorem}

\begin{proof}
In the following lemma, we denote by $P_\text{out} \subseteq P$ the subset of processes that output a value at \cref{line:async-all-uncond:out}.
Let us proceed by exhaustion.
\begin{itemize}
    \item In Case 1, we have $\bot \notin V$ and $t < n$.
    In this case, there is at least one correct process $p \in P$, and we must prove that \Cref{alg:async-all-uncond} implements the set of output sets $O = 2^V \setminus \{\varnothing\}$.
    Therefore, the set of forbidden output sets is $(2^{\{0,1\}} \setminus 2^V) \cup \{\varnothing\}$.
    
    \item In Case 2, we have $\bot \in V$ and $t \leq n$.
    In this case, all processes could crash, and we must prove that \Cref{alg:async-all-uncond} implements the set of output sets $O = 2^((V \setminus \{\bot\}))$.
    Therefore, the set of forbidden output sets is $2^{\{0,1\}} \setminus 2^{(V \setminus \{\bot\})}$.
\end{itemize}

We begin with general safety and completeness results that apply to both Cases 1 and 2.
\begin{enumerate}
    \item General-Safety: Impossibility of all sets in $2^{\{0,1\}} \setminus 2^{(V\setminus\{\bot\})}$ in Case 1 and 2.
    Let us consider a set $o \in 2^\{0,1\} \setminus 2^{(V\setminus\{\bot\})}$.
    Since $o$ is in $2^{\{0,1\}}$ but not in $2^{(V\setminus\{\bot\})}$, it means that $o$ contains values that are not in $V \setminus \{\bot\}$.
    But as processes only output values in $V \setminus \{\bot\}$ at \cref{line:async-all-uncond:out}, then $o$ is impossible.
    
    \item General-Completeness: Possibility of all sets in $2^{(V \setminus \{\bot\})} \setminus \{\varnothing\}$.
    Since $n \geq \max(\{|o| : o \in O\})$ and all failure patterns are possible, there is a set of executions where at least $\max(\{|o| : o \in O\})$ processes pass the condition at \cref{line:async-all-uncond:out} (either because $\bot \notin V$ or because all of these processes picked a non-$\bot$ value at \cref{line:async-all-uncond:pick}).
    In these executions, we have $|P_\text{out}| \geq \max(\{|o| : o \in O\})$.
    Among these executions, and for every $o \in 2^{(V \setminus \{\bot\})} \setminus \{\varnothing\}$, there is an execution where at least $|o|$ processes in $|P_\text{out}|$ pick and output each of the values in $o$ at \cref{line:async-all-uncond:out}, therefore producing output set $o$.
\end{enumerate}

\begin{algorithm}[tb]
\InstParam{set of values $V \subseteq \{0,1,\bot\}$.} \label{line:async-all-uncond:init}
\smallskip

\ProcCode{$p \in P$}{
    $v \gets \prpick(V)$; \label{line:async-all-uncond:pick} \\
    \lIf{$v \neq \bot$}{$\ooutput$ $v$;} \label{line:async-all-uncond:out}
}
\caption{Communication-less symmetric algorithm for lines 1-2 and 11-15 of \Cref{tab:characterization}.}
\label{alg:async-all-uncond}
\end{algorithm}

The safety of Case 2 follows from General-Safety.
For Case 1, we also need to show the impossibility of $\varnothing$.
Since $\bot \notin V$, some process $p$ is correct, must pass the condition at \cref{line:async-all-uncond:out}, and output a value at \cref{line:async-all-uncond:out}.
Therefore $\varnothing$ is impossible.

The completeness of Case 1 follows from General-Completeness.
For Case 2, we also need to show the possibility of $\varnothing$.
There is an execution where $P_\text{out} = \varnothing$ (because all processes either crashed or picked $\bot$ at \cref{line:async-all-uncond:pick}), therefore producing output set $\varnothing$.
\end{proof}

\subsection{Communication-Less Single-Output Algorithm (Lines 9-10)}
The algorithm of this section (\Cref{alg:async-single}) addresses lines 9 to 10 in \Cref{tab:characterization}.
It requires no communication, as it only involves one process outputting a pseudo-random value in isolation.
Therefore, this algorithm works both in the $\Sync$ and $\Async$ timing models.

This algorithm is instantiated with one parameter (\cref{line:async-single:instantiation}): a Boolean $\nooutput \in \{\ttrue, \ffalse\}$, determining if the empty $\varnothing$ output set can be produced or not.
At the initialization of the algorithm (\cref{line:async-single:pickp}), one special process $p \in P$ is chosen to be the only one that performs actions (\cref{line:async-single:pStart,line:async-single:output}), while all other processes do nothing.

Process $p$ first selects the set $V$ corresponding to the different output behaviors it can follow (\cref{line:async-single:initV}): $\{0,1\}$ if $\nooutput=\ffalse$, or $\{0,1,\bot\}$ if $\nooutput=\ttrue$.
Then, $p$ pseudo-randomly picks a value $v$ from $V$ (\cref{line:async-single:pickV}).
If $v \neq \bot$, then $p$ outputs $v$ at \cref{line:async-single:output}, otherwise $p$ does nothing.

\begin{algorithm}[tb]
\InstParam{Boolean $\nooutput \in \{\ttrue,\ffalse\}$.} \label{line:async-single:instantiation}
\smallskip

\Init{pick some $p \in P$.} \label{line:async-single:pickp}
\smallskip

\ProcCode{$p \in P$}{ \label{line:async-single:pStart}
    $V \gets \{0,1,\bot\}$ \textbf{if} $\nooutput=\ttrue$ \textbf{else} $\{0,1\}$; \label{line:async-single:initV} \\
    $v \geq \prpick(V)$; \label{line:async-single:pickV} \\
    \lIf{$v \neq \bot$}{$\ooutput$ $v$.} \label{line:async-single:output}
}
\caption{Communication-less algorithm for lines 9 and 10 of \Cref{tab:characterization}, with a single process outputting and assuming $n \geq 1$.}
\label{alg:async-single}
\end{algorithm}

\begin{theorem} \label{thm:async-single}
Under any timing model ($\Async$ or $\Sync$) and $n \geq 1$, \Cref{alg:async-single} instantiated with Boolean $\nooutput \in \{\ttrue,\ffalse\}$ implements the following sets of output sets:
$$ O = \begin{cases}
  \{\{0\},\{1\}\} &\text{if } \nooutput = \ffalse \text{ and } t = 0, \\
  \{\varnothing,\{0\},\{1\}\} &\text{if } \nooutput = \ttrue \text{ and } t \leq n.
\end{cases} $$
\end{theorem}

\begin{proof}
Let us proceed by exhaustion.
\begin{itemize}
    \item In Case 1, we have $\nooutput = \ffalse$ and $t = 0$.
    In this case, no process can crash, and we must prove that \Cref{alg:async-single} produces the set of output sets $\{\{0\},\{1\}\}$.
    Therefore, the set of forbidden output sets is $\{\varnothing,\{0,1\}\}$.
    
    \item In Case 2, we have $\nooutput = \ttrue$ and $t \leq 0$.
    In this case, all processes can crash, and we must prove that \Cref{alg:async-single} produces the set of output sets $\{\varnothing,\{0\},\{1\}\}$.
    Therefore, the set of forbidden output sets is $\{\{0,1\}\}$.
\end{itemize}

We begin with general safety and completeness results that apply to both Cases 1 and 2.
\begin{enumerate}
    \item General-Safety: Impossibility of $\{0,1\}$ in Cases 1 and 2.
    For an algorithm to produce two outputs in some execution, there must be at least two processes in some execution that output.
    But by construction, in \Cref{alg:async-single}, only process $p$ can output a value at \cref{line:async-single:pickp}.
    
    \item General-Completeness: Possibility of $\{0\}$ and $\{1\}$ in Cases 1 and 2.
    Since $n \geq \max(\{|o| : o \in O\})$ and all failure patterns are possible, there must be some executions in which process $p$ never crashes.
    Among these executions, there must be one in which $p$ picks $0$ (resp. $1$) at \cref{line:async-single:pickV} and outputs $0$ (resp. $1$) at \cref{line:async-single:output}, therefore producing the output set $\{0\}$ (resp. $\{1\}$).
\end{enumerate}

The safety of Case 2 follows from General-Safety.
For Case 1, we also need to show the impossibility of $\varnothing$. 
Since $\nooutput=\ffalse$ and no process can crash, $p$ will always pick some value at \cref{line:async-single:pickV} and output it at \cref{line:async-single:output}.

The completeness of Case 1 follows from General-Completeness.
For Case 2, we also need to show the possibility of $\varnothing$.
There is some execution where $p$ either crashes before outputting, or picks $\bot$ at \cref{line:async-single:pickV}.
In this execution, $p$ will not reach \cref{line:async-single:output}, therefore, producing the output set $\varnothing$.
\end{proof}

\subsection{Timing-Adaptive Algorithm (Lines 3-6)}
The algorithm of this section (\Cref{alg:async-all-cond}) addresses lines 3 to 6 in \Cref{tab:characterization}.
It is timing-adaptive as it provides different control flows depending on whether the system is synchronous or asynchronous.

This algorithm is instantiated by providing two parameters (\cref{line:async-all-cond:instant}): one value $v \in \{0,1\}$ and one Boolean $\nooutput \in \{\ttrue,\ffalse\}$.
Value $v$ is the \textit{default value}: it always has to be output by the system, as long as any value is output.
Boolean $\nooutput$ determines if the empty $\varnothing$ output set can be produced or not.
At the initialization of the algorithm, one special process $p \in P$ is chosen (\cref{line:async-all-cond:init}) to be the one that will guarantee that the opposite value of $v$ (\ie $1 \xor v$) can sometimes be output if $v$ has also been output.

\begin{algorithm}[tb]
\InstParam{value $v \in \{0,1\}$, $\nooutput \in \{\ttrue,\ffalse\}$.} \label{line:async-all-cond:instant}
\smallskip

\Init{pick some process $p \in P$.} \label{line:async-all-cond:init}
\smallskip

\ProcCode{$p' \in P \setminus \{p\}$}{ \label{line:async-all-cond:pNotLeader}
    \If{$\nooutput=\ffalse$ \cOr $\prpick(\{0,1\})=0$}{ \label{line:async-all-cond:PNL-Pick}
        $\ooutput$ $v$; \label{line:async-all-cond:PNL-Output} \\
        $\communicate$ $\outputi(v)$. \label{line:async-all-cond:PNL-Communicate}
    }
}
\smallskip

\ProcCode{$p$}{ \label{line:async-all-cond:leader}
    \If{$\nooutput=\ffalse$ \cOr $\prpick(\{0,1\})=0$}{ \label{line:async-all-cond:leaderPick}
        \lIf{the system is synchronous}{$\wait$ for the communication step of round $1$;} \label{line:async-all-cond:sync}
        \lElseIf{the system is asynchronous}{$\wait$ for a predefined local time;} \label{line:async-all-cond:async}
        \lIf{$p$ $\observed$ some $\outputi(v)$}{$\ooutput$ $\prpick(\{0,1\})$;} \label{line:async-all-cond:observeOtherOutput}
        \lElse{$\ooutput$ $v$.} \label{line:async-all-cond:noOtherOutput}
    }
}
\caption{Asymmetric algorithm for lines 3-6 of \Cref{tab:characterization}, assuming $n \geq 2$.
  Moreover, for lines 3 and 5 (which allow the $\varnothing$ output set), it assumes $t \leq n$, and for lines 4 and 6 (which forbid the $\varnothing$ output set), it assumes $t < n$.}
\label{alg:async-all-cond}
\end{algorithm}

The code of every process $p' \in P$ that is not $p$ is as follows (\cref{line:async-all-cond:pNotLeader,line:async-all-cond:PNL-Communicate}).
Process $p'$ passes the condition at \cref{line:async-all-cond:PNL-Pick} only if $\varnothing$ is a forbidden output set or if $p'$ pseudo-randomly flips the ``correct'' bit to $0$.
Then, $p'$ outputs $v$ (\cref{line:async-all-cond:PNL-Output}) and communicates some $\outputi(v)$ (\cref{line:async-all-cond:PNL-Communicate}).

The code of process $p$ is as follows (\cref{line:async-all-cond:leader,line:async-all-cond:noOtherOutput}).
Like with $p'$, process $p$ passes the condition at \cref{line:async-all-cond:leaderPick} only if $\varnothing$ is a forbidden output set or if $p$ pseudo-randomly flips the ``correct'' bit to $0$.
Then, $p$ waits for a given time, depending on the timing model: if the system is synchronous, $p$ waits for the communication step of synchronous round $1$ (\cref{line:async-all-cond:sync}), otherwise (\ie if the system is asynchronous) it waits until some predefined local time\footnote{
  In this case, the assumption that, in asynchrony, every communication-delay pattern can happen, helps us to guarantee that $p$ sometimes observes information before this local time, see \Cref{sec:formalization}.
} (\cref{line:async-all-cond:async}).
After this wait, $p$ outputs a random bit if it has observed some $\outputi(v)$ (\cref{line:async-all-cond:observeOtherOutput}), otherwise it outputs the default value $v$ (\cref{line:async-all-cond:noOtherOutput}).

\begin{theorem} \label{thm:async-all-cond}
Under any timing model ($\Async$ or $\Sync$) and $n \geq 2$, \Cref{alg:async-all-cond} instantiated with value $v \in \{0,1\}$ and Boolean $\nooutput \in \{\ttrue, \ffalse\}$ implements the following sets of output sets:
$$ O = \begin{cases}
  \{\{v\},\{v, 1 \xor v\}\} &\text{if } \nooutput = \ffalse \text{ and } t < n, \\
  \{\varnothing,\{v\},\{v, 1 \xor v\}\} &\text{if } \nooutput = \ttrue \text{ and } t \leq n.
\end{cases} $$
\end{theorem}

\begin{proof}
In the following lemma, we denote by $P_\text{out} \subseteq P$ the subset of processes that output a value at \cref{line:async-all-cond:PNL-Output}, \cref{line:async-all-cond:observeOtherOutput}, or \cref{line:async-all-cond:noOtherOutput}.
Let us proceed by exhaustion.
\begin{itemize}
    \item In Case 1, we have $\nooutput=\ffalse$ and $t < n$.
    In this case, there is at least one correct process $p \in P$, and we must prove that \Cref{alg:async-all-cond} produces the set of output sets $\{\{v\},\{v, 1 \xor v\}\} \setminus \{\varnothing\}$.
    Therefore, the set of forbidden output sets is $\{\varnothing, \{1 \xor v\}\}$.
    
    \item In Case 2, we have $\nooutput=\ttrue$ and $t \leq n$.
    In this case, all processes could crash, and we must prove that \Cref{alg:async-all-cond} produces the set of output sets $\{\varnothing,\{v\},\{v, 1 \xor v\}\}$.
    Therefore, the set of forbidden output sets is just $\{1 \xor v\}$.
\end{itemize}

We begin with general safety and completeness results that apply to both Cases 1 and 2.
\begin{enumerate}
    \item General-Safety: Impossibility of $\{1 \xor v\}$ in Cases 1 and 2.
    Any process $p' \in P\setminus\{p\}$ that outputs a value at \cref{line:async-all-cond:PNL-Output} will output $v$.
    Process $p$ outputs value $1 \xor v$ (with some probability) only if it has observed that at least one other process has output $v$ (\cref{line:async-all-cond:observeOtherOutput}). 
    C-Validity ensures that some process $p'$ first outputs $v$ and then communicates this information (\cref{line:async-all-cond:PNL-Output} and \cref{line:async-all-cond:PNL-Communicate}), and therefore that value $v$ has also been output.
    Thus, $o=\{1 \xor v\}$ is impossible.
    
    \item General-Completeness: Possibility of all sets in $\{\{v\},\{v, 1 \xor v\}\}$ in Cases 1 and 2. 
    We first argue the possibility of $\{v, 1 \xor v\}$. 
    Since $n \geq 2$ and all failure patterns (and communication-delay patterns in asynchrony) are possible, there is a set of executions where process $p$ (the selected one) and at least another process $p' \in P \setminus\{p\}$ do not crash and output a value.
    That is, process $p'$ passes the condition at \cref{line:async-all-cond:PNL-Pick} and therefore outputs $v$, and process $p$ passes the check at \cref{line:async-all-cond:leaderPick} and observes $\outputi(v)$ before reaching the condition at \cref{line:async-all-cond:observeOtherOutput}, and therefore outputs $\{1 \xor v\}$, with some non-null probability. 
    Observe that any other process in $P_\text{out} \setminus \{p\}$ outputs $v$, yielding the output set $\{v, 1 \xor v\}$.
    The possibility of $\{v\}$ follows similarly, for example, in executions that $p$ crashes or also outputs $v$.
\end{enumerate}

The safety of Case 2 follows from General-Safety.
For Case 1, we also need to show the impossibility of $\varnothing$.
Since $\nooutput=\ffalse$ and $t<n$, at least one process must pass the check in \cref{line:async-all-cond:PNL-Pick} or \cref{line:async-all-cond:leaderPick} and thus output a value (either at \cref{line:async-all-cond:PNL-Output} for $p'$, or \cref{line:async-all-cond:observeOtherOutput}, \cref{line:async-all-cond:noOtherOutput} for $p$).
Therefore, $\varnothing$ is impossible.

The completeness of Case 1 follows from General-Completeness.
For Case 2, we also need to show the possibility of $\varnothing$.
Consider an execution where all processes have either crashed before outputting, or picked $1$ at \cref{line:async-all-cond:PNL-Pick} or \cref{line:async-all-cond:leaderPick}, yielding the output set $\varnothing$.
\end{proof}

\subsection{Synchronous Binary Consensus Algorithm (Line 10)}

The algorithm of this section (\Cref{alg:sync-consensus}) addresses the $\Sync$ case of line 10 in \Cref{tab:characterization}.
Every process $p \in P$ starts by executing the \textit{communication step} of the first synchronous round (\cref{line:sync-consensus:pick,line:sync-consensus:comm}).
They first pick a pseudo-random value $v$ from the set $\{0,1\}$ at \cref{line:sync-consensus:pick}, and communicate a $\proposei(v)$ information at \cref{line:sync-consensus:comm}.
Then, every $p \in P$ moves to the \textit{computation step} of the first synchronous round (\cref{line:sync-consensus:out-0,line:sync-consensus:out-1}).
At this stage, we can show that all non-crashed processes see the same set of $\proposei(v)$ information.
Process $p$ output $0$ if it has observed any $\proposei(0)$ information during round $1$ (\cref{line:sync-consensus:out-0}), otherwise, $p$ output $1$ (\cref{line:sync-consensus:out-1}).

\begin{theorem} \label{thm:sync-consensus}
Under $\Sync$ and $t < n \geq 1$, \Cref{alg:sync-consensus} implements the set of output sets $O=\{\{0\},\{1\}\}$.
\end{theorem}

\begin{proof}
Under $\Sync$ and $t < n \geq 1$, there is at least $n-t \geq 1$ correct process $p$.
The safety of \Cref{alg:sync-consensus} is proved in the following.
\begin{itemize}
    \item Impossibility of $\varnothing$.
    Process $p$ necessarily outputs a value at \cref{line:sync-consensus:out-0} or \cref{line:sync-consensus:out-1}, therefore $\varnothing$ is impossible.
    
    \item Impossibility of $\{0,1\}$.
    Let us denote by $P_\text{obs} \subseteq P$ the set of processes that observe some $\outputi(v)$ at the end of round $1$.
    By C-Synchrony, all processes in $P_\text{obs}$ must have observed the same set of $\proposei(v)$ communicated at the beginning of round $1$, and they must have therefore taken the same branch at \cref{line:sync-consensus:out-0} or \cref{line:sync-consensus:out-1}, thus the output cannot contain 2 distinct values.
\end{itemize}

\begin{algorithm}[tb]
\ProcCode{$p \in P$ \textbf{in synchronous round} $R=1$}{
    \commStep
    $v \gets \prpick(\{0,1\})$; \label{line:sync-consensus:pick} \\
    $\communicate$ $\proposei(v)$; \label{line:sync-consensus:comm} \\
    \compStep
    \lIf{$p$ $\observed$ some $\proposei(0)$}{$\ooutput$ $0$;} \label{line:sync-consensus:out-0}
    \lElse{$\ooutput$ $1$. \Comment*[f]{\ie $p$ observed only \textup{$\proposei(1)$}'s}} \label{line:sync-consensus:out-1}
}
\caption{Synchronous algorithm for the $\Sync$ case of line 10 of \Cref{tab:characterization}, assuming $t < n \geq 1$.}
\label{alg:sync-consensus}
\end{algorithm}

The completeness of \Cref{alg:sync-consensus} is proved in the following.
\begin{itemize}
    \item Possibility of $\{0\}$.
    There is an execution where $p$ picks value $v=0$ at \cref{line:sync-consensus:pick}, communicates $\outputi(0)$ at \cref{line:sync-consensus:comm}, and observes its own $\outputi(0)$ and outputs $0$ at \cref{line:sync-consensus:out-0}.
    By the impossibility of $\{0,1\}$, this execution will produce output set $\{0\}$.
    
    \item Possibility of $\{1\}$.
    Let us denote by $P_\text{com} \subseteq P$ the set of processes that communicate some $\outputi(v)$ at the beginning of round $1$.
    There is an execution where all processes in $P_\text{com}$ have picked value $v=1$ at \cref{line:sync-consensus:pick}, and therefore communicate $\proposei(1)$ at \cref{line:sync-consensus:comm}.
    By C-Local-Termination, C-Global-Termination, and C-Synchrony, process $p$ must observe only $\proposei(1)$ information at the end of round $1$, and it must therefore output $1$ at \cref{line:sync-consensus:out-1}.
    By the impossibility of $\{0,1\}$, this execution will produce output set $\{1\}$. \qedhere
\end{itemize}
\end{proof}

\end{document}